\documentclass{article}

\usepackage[utf8]{inputenc}
\usepackage{parskip}        
\usepackage{amsmath,amsthm}
\usepackage{amsfonts}
\usepackage{hyperref}
\usepackage{xcolor}
\usepackage[letterpaper,margin=1.00in]{geometry}
\usepackage[ruled,vlined]{algorithm2e}
\usepackage{algorithmic,multicol}

\makeatletter
\newtheorem*{rep@theorem}{\rep@title}
\newcommand{\newreptheorem}[2]{%
\newenvironment{rep#1}[1]{%
 \def\rep@title{#2 \ref{##1}}%
 \begin{rep@theorem}}%
 {\end{rep@theorem}}}
\makeatother

\newtheorem*{theorem*}{Theorem}

\newtheorem{theorem}{Theorem}
\newtheorem{lemma}{Lemma}

\newtheorem{definition}{Definition}

\newcommand{\E}{\mathbb{E}}

\newcommand{\ceil}[1]{\left\lceil #1 \right\rceil}
\newcommand{\dk}[1]{#1}
\newcommand{\jo}[1]{#1}

\newcommand{\remove}[1]{}

\newcommand{\CF}{f}
\newcommand{\CA}{g}
\newcommand{\dmin}{d_{\min}}
\newcommand{\dmax}{d_{\max}}
\newcommand{\wellconnected}{well-connected{}}

\newcommand{\cA}{\mathcal{A}}

\newcommand{\cI}{\mathcal{I}}

\newcommand{\cL}{\mathcal{L}}

\title{Deterministic Fault-Tolerant Local Load Balancing
and its Applications against Adaptive Adversaries}
\author{	Dariusz R. Kowalski \footnotemark[1]
		\and
		Jan Olkowski \footnotemark[2]}

\footnotetext[1]{	School of Computer and Cyber Sciences, Augusta University, Augusta, Georgia, USA.}

\footnotetext[2]{Independent.}
\date{}

\begin{document}

\maketitle

\begin{abstract}
Load balancing is among the basic primitives in distributed computing. In this paper, we consider this problem when executed locally on a network with nodes prone to failures. We show that there exist lightweight network topologies that are immune to message delivery failures incurred by (at most) a constant fraction of all nodes. More precisely, we design a novel deterministic fault-tolerant local load balancing (LLB) algorithm, which, similarly to their classical counterparts working in fault-free networks, has a relatively simple structure and guarantees exponentially fast convergence to the average value despite crash and omission failures. 

As the second part of our contribution, we show three applications of the newly developed fault-tolerant local load balancing protocol. We give a randomized consensus algorithm, working against $t < n / 3$ crash failures, that improves over the best-known consensus solution by Hajiaghayi et al. (STOC'22) with respect to communication complexity, yet with an arguable simpler technique of combining a randomly and locally selected virtual communication graph with a deterministic fault-tolerant local load balancing on this graph. 

We also give a new solution for consensus for networks with omission failures. Our solution works against $t < \frac{n}{C\log{n} (\log\log n)^2}$ omissions, for some constant $C$, is nearly optimal in terms of time complexity, but most notably -- it has communication complexity $O((t^2 + n)\text{ polylog } {n})$, matching, within a polylogarithmic factor, the lower bound by Abraham et. al. (PODC'19) with respect to both terms depending on $t$ and $n$. Ours is the first algorithm in the literature that is simultaneously nearly optimal, in terms of $n,t$, with respect to both complexity measures, against the adaptive omission-causing adversary.
Finally, we show an application of load balancing to the problem of almost everywhere counting in networks prone to crash or omission failures. Our techniques improve the communication complexity of these applications, compared to the best-known algorithms, by at least a polylogarithmic factor, and in case of consensus against omission failures -- even by a polynomial.
\end{abstract}

\thispagestyle{empty}
\setcounter{page}{0}

\newpage

\section{Introduction}

A 
fault-tolerant {\bf\em local load balancing (LLB)} is a distributed problem in which every process (also called a vertex or a node) starts with an input value in $[0, 1]$, and 
at the end 
there exists a subset of processes, of size $n - O(t)$, such that the values stored by these processes 
differ by at most $\tilde O(t / n)$ from the mean $\mu$ of all the input values, where $n$ is the number of all processes and $t$ is the number of faulty ones.\footnote{Notation $\tilde O(\cdot)$ hides poly-logarithmic factors.}

{\bf Message-passing system.}
We consider a message-passing fully connected distributed system of $n$ processes. The processes know $n$ (or its linear estimate).
They work synchronously in rounds. Each of them knows only its local numbering of communication ports, but the port numbers do not automatically provide information about the identity of the other side of the adjacent link. We also do not require the processes to 
have a unique ID, i.e., they could be anonymous. Messages sent via point-to-point links are of size $O(\log{n})$ each, but a process may send different messages via different ports in a round (multicast operation).

Our LLB algorithm can work efficiently not only in fully-connected message-passing systems, as described above, but also in a class of well-connected network topologies (see Def.~\ref{def:well-connected}). 
The considered applications, however, require full connectivity of the communication network.

{\bf Process failures.}
We consider two standard types of process' failures:
crashes 
and omissions. 
Suppose a {\em crash failure} of a process occurs in a round of computation. In that case, 
only some arbitrary subset of links adjacent to this process succeeds in delivering messages from the process in this round, and starting from the next round, the process halts its computation entirely: no messages can be sent by nor delivered to this process, and the process is excluded from any further consideration.
An {\em omission failure} of a process means that its adjacent links could occasionally omit some in-coming and out-going messages, starting from the round in~which~this~process~becomes~faulty. 

{\bf Adversary.}
In any type of failures, we assume that an {\em adaptive full-information adversary} controls the pattern of these failures. More specifically, the adversary knows the protocol and can control chosen processes, up to $t$ of them, together with their adjacent links, in an online fashion based on the history of the computation in the system. If processes use randomness, we assume that the adversary can see the entire system's state at any moment of the execution, but cannot foresee future random bits to be provided to the processes (if~any).


Load balancing is a widely applicable primitive -- we demonstrate efficient applications of our fault-tolerant LLB algorithm to the problems of counting and consensus under crash and~omission~failures. 

{\bf Counting.}
In the {\bf\em $k$-almost-everywhere counting problem}, there is a subset of processes with a raised flag, and each process knows in the beginning only whether its flag is raised or not. \\
{\bf\em Correctness:} At the end, there 
is
a subset of processes, of size at least $k$, such that the counting algorithm at each process in this subset returns the number of processes with raised flag, with an additive accuracy~of~$\tilde{O}(t)$. 

{\bf Consensus.}
In the {\bf\em binary consensus problem} each process is given an input value in $\{0, 1\}$ and the goal of every correct process is to decide on some of the processes' input, adhering to the three requirements: \\
{\bf\em Validity:} Only a value among the initial ones may be decided upon.\\
{\bf\em Agreement:} No two non-faulty processes decide on different values.\\
{\bf\em Termination:} Each process eventually decides on some value unless it is faulty.

We consider randomized solutions to the above problems, and thus, we require that all problem-specific conditions hold with high probability. We say that an event holds \textit{with high probability (whp)} if there exists an absolute constant $C \ge 1$ such that the probability of this event is at least~$1-\frac{1}{n^{C}}$.

\subsection*{Our contributions}

Our main technical contribution is a \underline{deterministic fault-tolerant 
LLB}
algorithm $\textsc{FaultTolerantLLB}$, which works under both \underline{crash and omission failures}. The algorithm executed by a process $v$, takes as input 
a subset of ports
and the load of a node corresponding to the process executing the algorithm, and returns a pair $x(v), \texttt{type}$, where $x(v)$ is the final load at process $v$. If the conceptual communication graph/network $G$, formed by edges corresponding to the input ports at nodes (i.e., edge $\{v,w\}$ means one of the input ports in $v$ leads to $w$ and vice versa), is  
$(d_{\min}, d_{\max})$-\wellconnected, 
i.e., it satisfies certain connectivity properties and each node has degree between 
$d_{\min}$ and $d_{\max}$ 
(see Definition~\ref{def:well-connected}), then our algorithm guarantees the following (here we present a simplified version of Theorem~\ref{thm:load-balancing} for suitable constants $d_{\min},d_{\max}$):

\begin{theorem*}[Theorem~\ref{thm:load-balancing} in Section~\ref{sec:load-balancing}, version for constant $d_{\min},d_{\max}$]
Let $G$ be a $(d_{\min}, d_{\max})$-\wellconnected\space  graph. 
The algorithm \textsc{FaultTolerantLLB} executed on the graph $G$ under at most $t$ crash or omission failures, achieves the following guarantees:\\
\vspace{-3mm}\\
\noindent \textit{$(i)$} it terminates in $O\left(\tau_1 + \tau_2\right)$ rounds using $O\left(\left(\tau_1 + \tau_2
\right)|M|\right)$ communication bits per node, where $|M|$ is the size of machine word and $\tau_1,\tau_2 = O(\log{n})$;\\ 
\vspace{-3mm}\\ 
\noindent \textit{$(ii)$} the first element of the returned pair is always between the largest and the smallest input value;\\
\vspace{-3mm}\\ 
\noindent \textit{$(iii)$} if $t < C\cdot n$, for some constant $C$ depending on $d_{\min},d_{\max}$, then there exists a set of nodes $A$, of size at least $n - \frac{3}{2} t$, such that for every $v \in A$, every returned pair $x(v), \texttt{type}$ satisfies $\texttt{type} = active$;\\
\vspace{-3mm}\\ 
\noindent \textit{$(iv)$} let $\varepsilon \in [0,1]$ be such that $t < \frac{\varepsilon}{3\tau_1}\cdot C' \cdot n$, for some $C'$ depending only on $d_{\min},d_{\max}$, and let $\mu$ denote the mean of the input values; Then, for every node $v \in A$ it holds that 
\vspace*{-2ex}
\[
x(v) \in [\mu - \varepsilon, \mu + \varepsilon]
\ .
\]
\end{theorem*}

\remove{
\begin{theorem*}[Theorem~\ref{thm:load-balancing} in Section~\ref{sec:load-balancing}]
Let $\varepsilon$ be a constant in $[\frac{2}{n}, 1]$, $G$ be $(d, \gamma)$-\wellconnected and let $t \le \alpha n$ denote the upper bound on the number of faulty nodes, where $\alpha \in [0, 1]$ is some absolute constant. 
The algorithm $\textsc{FaultTolerantLLB}$ executed on the graph $G$ guarantees the following:\\
\vspace{-3mm}\\
\noindent \textit{$(i)$} it terminates in $O\left(\tau + \log{n}\right)$ rounds using $O\left(d\tau |M|\right)$ communication bits per node, where $|M|$ is the size of \dk{machine word???} and $\tau = 4\gamma^2 \log{n}$,\\ 
\vspace{-3mm}\\ 
\noindent \textit{$(ii)$} there exists an absolute constant $\beta \ge 1$ and a set of nodes $A$ of size at least $n - \beta t$ such that for each $v \in A$, the returned pair $x(v), \texttt{type}$ satisfies $\texttt{type} = active$ and $x(v)$ is between the smallest and the largest input value,\\
\vspace{-3mm}\\ 
\noindent \textit{$(iii)$} if $t <\min\left( \frac{\alpha n}{ \left(2\tau / \varepsilon + 1 \right)}, \frac{n}{ 15 \gamma^2 \beta \left( \tau / \varepsilon + 1 \right)} \right)$, \dk{for some $\epsilon ???$,} then for every node $v \in A$ it holds additionally that  $x(v) \in [\mu - \varepsilon, \mu + \varepsilon]$, where $\mu$ is the mean of all the input values.
\end{theorem*}
}

\vspace*{-1ex}
The main technical advancement is in combining specific LLB formula in line~\ref{line:main-for-end}, applied $\tau_1=O(\log n)$ times in the main loop of the LLB Algorithm~\ref{alg:load-balancing}, with the fault-tolerant mechanism of fixing outliers that follows (see its description in Algorithm~\ref{alg:fixoutliers}).
The latter mechanism is checking in $\tau_2=O(\log n)$ rounds whether enough number of still ``non-suspected'' neighbors remain active -- positive check implicitly confirms that the preceding LLB procedure in $\tau_1$-hop neighborhood was executed by a large-enough non-suspected subgraph of the original communication graph and thus the returned value is a good estimate of the mean value.

We could apply our LLB algorithm on pre-determined or random graphs (i.e., on pre-determined or randomly selected input ports), 
to get deterministic or randomized version of LLB, respectively. 
Hence, it could be applied to improve performance of both deterministic and randomized solutions to various 
distributed computing 
problems. 
To demonstrate its efficiency, we 
apply 
it to solve
\underline{almost-everywhere counting} and \underline{consensus} under \underline{crash and omission failures}.

Applying the \textsc{FaultTolerantLLB} algorithm to the problem of almost-everywhere counting yields: 

\begin{theorem*}[Theorem~\ref{thm:counting} in Section~\ref{sec:applications}]
For any number of either crash or omission failures 
\dk{$t \le \frac{n}{\log{n}}$,} 
there is an algorithm that solves $(n-3t)$-almost-everywhere counting problem in $O(\log{n})$ rounds using $O(n\log^3{n})$ communication bits, whp.
\end{theorem*}

The new counting algorithm can be applied to solve the consensus problem for crash failures.

\begin{theorem*}[Theorem~\ref{thm:simple-cons-crashes} in Section~\ref{sec:app-cons-crashes}]
There exists an algorithm that solves consensus against 
$t<n/3$
crashes in $O(\sqrt{n}\log^{3/2}{n})$ rounds using 
$O\left(n^{3/2}\log^{5/2}(n)\left(\log\log{n}\right)^2\right)$communication bits, whp.
\end{theorem*}

This result improves 
the communication complexity of the previous most efficient algorithm from~\cite{hajiaghayi2022improved}, Theorem 2, by multiplicative factor $O(\log^{3/2}{n})$. We note, however, that the algorithm in~\cite{hajiaghayi2022improved} achieves correctness with probability $1$ and has  round complexity lower by factor $O(\log^{1/2}{n})$. Nevertheless, the technique proposed in this paper is much simpler and also more versatile -- our model does not require knowing the identifiers, 
which by itself requires $\Omega(n^2)$ of total communication if not provided 
as the part of the code.

Finally, 
our technique of fault-tolerant LLB also provides a lightweight, simple and efficient consensus algorithm against omission failures, see the theorem below. 
Although our algorithm tolerates 
$O(\frac{n}{\log {n} (\log\log n)^2})$ 
faulty processes, the result holds against the adaptive, full-power adversary. Our algorithm is time-efficient, in the sense that it is polylogarithmically close to the lower bound $\Omega\left(\frac{t}{\sqrt{n\log n}}\right)$ by Bar-Joseph and Ben-Or~\cite{Bar-JosephB98}. 
Since  
$O\left(t^2 \log^2 n\left(\log\log{n}\right)^4 \allowbreak + \left(t\sqrt{n} + n\right)\log^{3}{n}\left(\log\log{n}\right)^2 \right)\subseteq O((t^2+n)\text{ polylog } n)$, the communication complexity of our algorithm is also polylogarithmically close to the lower bound $\Omega(t^2+n)$ for randomized algorithms, by Abraham et al.~\cite{AbrahamCDNPRS19}.
The best known time-efficient algorithmic solution, by Hajiaghayi et al.~\cite{hajiaghayi2024nearly}, requires $\Theta(n^2 \text{ polylog }n)$ communication bits, which is polynomially worse than ours for $t=O(n^{1-\epsilon})$.

\begin{theorem*}[Theorem~\ref{thm:omissions} in Section~\ref{sec:app-consensus-omissions}]
There exists an algorithm that solves consensus against 
$t < \frac{n}{C\log{n}\left(\log\log{n}\right)^2}$ omission failures in $O\left(\frac{t\log^2{n}}{\sqrt{n}} + \log^2{n} \right)$ rounds using $O\left(t^2 \log^2 n\left(\log\log{n}\right)^4 \allowbreak + \left(t\sqrt{n} + n\right)\log^{3}{n}\left(\log\log{n}\right)^2 \right)$ communication bits, for some constant $C$, whp.
\end{theorem*}

\subsection*{Related work}

A problem in which \textit{all} processes calculate a function of certain parameters that are held at individual nodes is closely related to the local load balancing problem considered in this paper. For example, when all nodes calculate the average of these initial values, they are said to reach average consensus. Average consensus and, more generally, distributed function calculation, especially based on diffusion algorithms, have received a tremendous amount of attention by many communities, including the control community, multi-agent systems, signal-processing and sensor networks community, as well as theoretical distributed computing~\cite{erseghe2011fast, olfati2004consensus, pilloni2015consensus, sardellitti2009fast, tu2012diffusion, yang2023adaptive}. Importantly, usually in this line of work, processes are assumed to be \textit{always} correct and the emphasis is usually on: the pace of asymptotic convergence with respect to network's topology~\cite{kashyap2007quantized, erseghe2011fast, sardellitti2009fast}, and the issues of finite time completion~\cite{sundaram2007finite, wang2010finite} or the quantized transmissions~\cite{kashyap2007quantized}. Reference~\cite{hadjicostis2018distributed} discusses several applications of distributed average consensus. 

Some important derivations of the average consensus problem considered variations where processes are allowed to exhibit errors~\cite{hadjicostis2022trustworthy}. Nevertheless, such settings usually divert from the classical crash/omission/Byzantine failure model by that faulty processes can only manipulate received values, not the communication topology.

Load-balancing can be understood as a problem where parties are given a potentially uneven assignment of “work” which they must then attempt to spread evenly. \cite{DinitzFGN17, AmelinaFJV15, KowalskiM21} studied this problem when the underlying topology has a stochastic nature and derive asymptomatic convergence based on the distribution of edges in the topology graph. In~\cite{KowalskiM20}, a deterministic load balancing was used as part of counting in dynamic networks, however there was an additional assumption that there is a leader in the network -- such an assumption does not stand if faulty processes are considered (which is the case in our work). 

We apply LLB to improve efficiency of consensus. 
Consensus problem under crash failures has been widely studied. The best deterministic solution uses time $O(t)$ and $O(n+t\log n)$ communication bits~\cite{ChlebusKO23}, while the corresponding lower bounds $\Omega(t)$ and $\Omega(n)$ were proved, respectively, in~\cite{FischerL82} and~\cite{AmdurWH92}.
Randomized consensus for crash failures against a powerful adaptive adversary is less studied, and the recent result in \cite{hajiaghayi2022improved} suggests that there could be a trade-off between time complexity, which is $\Omega\left(\frac{t}{\sqrt{n\log n}}\right)$ \cite{Bar-JosephB98}, and communication complexity, which for linear-time solution could be very small (see the deterministic case) but for fast solutions the best known is $O(n^{3/2}\text{polylog }n)$ \cite{hajiaghayi2022improved}. 
The abovementioned lower bounds for time complexities naturally extend to more severe omission failures, while a stronger lower bound on communication $\Omega(t^2+n)$ was proved by Abraham et al.~\cite{AbrahamCDNPRS19}.

Other than consensus, similar set of techniques and approaches has been applied to gossip - a problem where all processes start with an initial value and over time all correct processes must learn initial value of other correct processes (but can arbitrarily differ on knowledge of incorrect processes)~\cite{alistarh2010efficient, aysal2009broadcast, cai2011quantized, dimakis2010gossip}. 

Another derivative of consensus is \textit{counting} -- a problem in which the correct parties' goal is to give an estimate of their number in finite time. \cite{ChatterjeePR22} studied Byzantine resilient counting without assuming any approximate knowledge of the number of processes in the system.  Counting in networks with weaker fault models was also studied. The work of~\cite{KowalskiM21} considers the case when faults are oblivious (i.e., set up before the algorithm could mitigate them). More restricted scenario, such as presence of a leader or a small number of failures, was studied in~\cite{KowalskiM20,ChandranCGGOZ15,ChatterjeePR22}.
The counting problem under process crashes was studied in~\cite{hajiaghayi2022improved}.

In this context, and to the best of our knowledge, ours is the first paper studying theoretically the problem of local load balancing, in the averaging sense, against an \textit{adaptive, full-information} adversary causing processes' crashes and/or message omissions. Rather than on the stochastic nature of the problem, we focus on combinatorial properties of graphs that make the fault-tolerant averaging / load-balancing possible. We also do not assume any  additional hints to correct processes regarding the state of the faulty ones.

\vspace*{-2ex}
\paragraph{Organization of the paper.}
Section~\ref{sec:preliminaries} contains necessary background and notation from the combinatorial, spectral and random graph theory, and selected fault-tolerant properties of graphs. Section~\ref{sec:load-balancing} presents the main result -- local load balancing algorithm $\textsc{FaultTolerantLLB}$ and its analysis. 
Applications to counting, and to consensus problems tolerating crash and omission failures, are given in Sections~\ref{sec:applications},~\ref{sec:app-cons-crashes} and~\ref{sec:app-consensus-omissions}, resp. 
Omitted technical proofs from Sections~\ref{sec:preliminaries},~\ref{sec:load-balancing} and~\ref{sec:applications} are provided in Sections~\ref{sec:proofs-preliminaries},~\ref{sec:appendix-lb} and~\ref{sec:app-consensus-crash-analysis},
resp.

\section{Preliminaries}
\label{sec:preliminaries}







\paragraph{\bf Graph notation.} 
We introduce basic definitions from graph theory and spectral graph theory. Let $G=(V,E)$ denote an undirected graph. We use $N_{G}(v)$ to denote the neighbors of $v$ in $G$ and $deg(v)$ to indicate the degree of $v$.
Let $W\subseteq V$ be a set of nodes in~$G$.
We say that an edge $(v,w)$ of $G$ is \emph{internal for~$W$} if $v$ and~$w$ are both in~$W$. We denote $E(W)$, the number of internal edges for $W$.
We say that an edge $(v,w)$ of $G$ \emph{connects the sets~$W_1$ and $W_2$},
or \emph{is between $W_1$ and $W_2$}, for any disjoint subsets $W_1$ and~$W_2$ of~$V$, if one of its ends is in~$W_1$ and the other in~$W_2$.
We denote $\partial(W_1)$, the set of edges between $W_1$ and $V \setminus W_1$. We denote by $vol(W_1)$ the sum of degrees of all vertices $W_1$. 

\vspace*{-2ex}
\paragraph{\bf Spectral graph theory.}
We define the adjacency matrix $A(G)$ (in general, we omit the reference $(G)$ in the matrix notation if $G$ is known from the context) as a $|V| \times |V|$ matrix with $A_{i,j} = 1$ if there is an edge between $i$ and $j$, $i \neq j \in V$, and with $A_{i, j} = 0$ otherwise. We denote $D(G)$, the diagonal matrix with the degrees of vertices of $G$ written on the main diagonal. 
For a diagonal matrix $X$ with only positive entries, we define matrix $X^{-1/2}$ as a diagonal matrix with inverse square roots of the corresponding diagonal entries of $X$ on its diagonal (and zeros outside the diagonal). 
The normalized adjacency matrix is then defined as 
$\cA = D^{-1/2} A D^{-1/2}$.
The normalized Laplacian matrix of the graph $G$ is defined as 
$\cL = I - \cA$,
where $I$ is the identity matrix.

The eigenvalues of the matrix $\cL$ are closely related to combinatorial properties of $G$. In fact, $\cL$ has $n$ real eigenvalues. We use $\lambda_{1} \le \lambda_{2} \le \ldots \le \lambda_{n}$ to denote the eigenvalues of $\cL$. It is well-known that eigenvalues of the normalized Laplacian matrix are non-negative and do not exceed, c.f. Section 4.3.1 in~\cite{hoory2006expander}.

\vspace*{-1ex}
\paragraph{\bf Fault-tolerant properties of graphs.}

The connectivity properties of $G$ are closely related with the magnitude of the second smallest eigenvalue of $\cL$. 

\begin{lemma}\label{cor:edge-density}
For any graph $G$ with second eigenvalue $\lambda_2$ and any subset of nodes $W$,
\[
E(W) \le \frac{1}{2}\cdot vol(W)\left(1 - \lambda_{2}\cdot\left(1 - \frac{vol(W)}{vol(G)} \right) \right)
\ .
\]
\end{lemma}

For this reason, we will use the following definition of a \wellconnected graph.
\begin{definition}
\label{def:well-connected}
A graph $G$ is called $(d_{\min}, d_{\max})$-\textit{\wellconnected} 
if and only if\\ 
\noindent $(i)$ the degree of any vertex $v$ of $G$, satisfies $deg(v) \in [d_{\min}, d_{\max}]$, and\\ 
\noindent $(ii)$ $\jo{\lambda_2 \ge 1 - \frac{1}{10\log\log{n}}}$, where $n$ is the number of vertices of $G$.
\end{definition}

An example of a $(d,d)$-\wellconnected{} graph is an $(n, d, \lambda)$-expander, for $d = \omega_{n}(1), \lambda = o(\log\log{d})$.
Recall that an $(n, d, \lambda)$-expander is a $d$-regular graph of $n$ nodes and of the second eigenvalue of the adjacency matrix equal to $\lambda$.
\jo{ We note that however similar in flavor, the definition of a $(d_{\min}, d_{\max})$-\wellconnected{} graph is broader, as it allows irregular graphs. The tolerance of the irregularity in the degrees of vertices will become crucial in later applications.}

In the case of expander graphs, it was observed
in~\cite{upfal1992tolerating}, that removing certain linear fraction of nodes of an expander graph leaves a densely connected ``core subgraph'' in the remaining part of the graph. Using analogical reasoning, we recover this observation for the class of \wellconnected{}~graphs.


\vspace*{-2ex}
\paragraph{Formal construction of core subgraph.} Let $G$ be a $(d_{\min}, d_{\max})$-\wellconnected{} graph, and let us consider the following scenario. Let $F$ be a given set of vertices to be removed from $G$. In later applications, the set $F$ will often coincide with processes that are flawed for the purpose of our algorithms. 
Consider an inductive construction of a sequence of sets $W_0, W_1, \ldots$ which starts with the set $W_0 := F$ and in every step $i \ge 1$, $W_i$ is formed by enlarging $W_{i-1}$ by any vertex from $G$ that has less than $\phi d_{\min}$ neighbors in $G \setminus W_{i - 1}$. Since $W_0 \subseteq W_1 \subseteq W_2 \ldots$ and there are $n$ vertices, the sequence of sets has a fixed point, which we denote~$W_{k}$. Let $G'$ be a graph induced by $V \setminus W_{k}$. We note that every vertex in $G'$  has a degree~at~least~$\phi d_{\min}$.

\begin{lemma}\label{lem:high-degree-core}
%
Let $G = (V, E)$ be a $(d_{\min}, d_{\max})$-\wellconnected~graph, $F$ be a subset of $V$ such that $|F| < \alpha |V|$, for $\alpha \in (0, 1)$, and furthermore assume that the inequality relates $\alpha$ and $\phi$:
$    \alpha < (1-\phi)\frac{40}{27}\frac{d^2_{\min}}{d^2_{\max}} - \frac{2}{9}\frac{d_{\min}}{d_{\max}}$.
%
%
Then the size of the set $W_{k}$ in the above construction of the core subgraph is at most $\ceil{\frac{3}{2} |F|}$.
%
\end{lemma}

\vspace*{-1ex}
\paragraph{\bf Random graphs.} 
We use the Erdős–Rényi model of random graphs. In this model, each edge appears independently with probability $p$. The distribution of graphs with $n$ vertices in the Erdős–Rényi model is denoted $G(n,p)$. 

\begin{lemma}\label{cor:erdos-renyi-lap}
Let $G$ be a random graph drawn from $G(n,p)$, for some parameter $p$. There exists a constant $C$ such that if $p \ge \jo{C}\cdot \frac{\log{n}(\log\log{n})^{2}}{n - 1}$, then the following holds with probability at least $1 - \frac{1}{n^2}$: \\
(a) degree of every vertex of $G$ is in the interval $\left[p\cdot (n-1)\left(1 - \frac{1}{20\log\log{n}}\right), p\cdot (n-1)\left(1 + \frac{1}{20\log\log{n}}\right)\right]$,\\ 
(b) \jo{ $\lambda_2(G) \ge 1 - \frac{1}{10 \log\log{n}}$ }. 

\end{lemma}

\section{Fault-tolerant local load balancing}
\label{sec:load-balancing}

We present a 
novel deterministic fault-tolerant
local load balancing algorithm, \textsc{FaultTolerantLLB}, which pseudocode is given in Algorithm~\ref{alg:load-balancing}. We prove that it can converge in a distributed system in which processes are prone to an arbitrary pattern of adaptive failures. Recall that $n$ denotes the number of processes in the distributed system and $t < n/10$ \footnote{We do not aim here at optimizing the linear factor in the upper bound on the number of faulty nodes.} denotes an upper bound on the number of erroneous processes.

\IncMargin{1.5em} 
\vspace*{-1ex}
\begin{algorithm}
\LinesNumbered
\SetKwInOut{Input}{input}
\Input{$(d_{\min}, d_{\max})$-\wellconnected{} graph $G$, $v$, $b_{v}$}
$x_{0}(v) \leftarrow b_v$\;
\For{$i \leftarrow 1 \text{ to } \tau_1 \leftarrow 32d^2_{\max} / d^2_{\min} \log{n}$\label{line:main-for-start}} {
send $x_{i - 1}(v)$ to every vertex in $N_{G}(v)$\;
let $N_{i}(v)$ be the set of vertices from which $v$ received a message in the previous round\;
$x_{i}(v) \leftarrow \left( \sum_{u \in N_{i}(v)} \frac{1}{2\dmax}x_{i - 1}(u) \right) + \frac{2\dmax - |N_{i}(v)|}{2\dmax} x_{i - 1}(v)$\;\label{line:main-for-end}
}
$x(v), \texttt{type} \leftarrow$ \textsc{FixOutliers}$(G, v, x_{\tau_1}(v))$\;
\Return{$x(v), \texttt{type}$}
\caption{\textsc{FaultTolerantLLB}
\label{alg:load-balancing}
}
\end{algorithm}

\vspace*{-3ex}
\begin{algorithm}
\SetKwInOut{Input}{input}
\LinesNumbered
\Input{$(d_{\min}, d_{\max})$-\wellconnected{} graph $G$, $v, x_0(v)$}
$\texttt{type} \leftarrow active$\;
\For{$i \leftarrow 1 \text{ to } \tau_2 \leftarrow  \log{n} / \log\left(\frac{34}{15} - \frac{4d_{\min}}{3d_{\max}}\right)$} {
send $x_{i - 1}(v)$ to every vertex in $N_{G_1}(v)$\;
let $N_{i}(v)$ be the set of vertices from which $v$ received a message in the previous round\;
\label{line:many-neighbors}\If{$|N_{i}(v)| < \frac{2}{3}d_{\min}$\label{line:if-silent}} {
$\texttt{type} \leftarrow silent$\label{line:silent}\; 
$v$ stops communication until the loop finishes\;
}
\Else {
$x_{i}(v) \leftarrow \text{median}\{ x_{i-1}(u) : u \in N_{i}(v)$\}\;
}
}
\Return($x_{\log_{14/15}{n} + 1}(v), \texttt{type}$)
\caption{\textsc{FixOutliers}\label{alg:fixoutliers}}
\end{algorithm}
The algorithm takes as an input a $(d_{\min}, d_{\max})$-\wellconnected{} graph $G$ and the input load of a process that executes the algorithm. We identify the processes with the vertices of $G$. For the ease of presentation and analysis, we assume that the knowledge of $G$ given to a node consists of a representation of the entire graph, the parameters $(d_{\min}, d_{\max})$, and the identifier of the vertex $v$ executing this algorithm. 
However, the minimal required knowledge for $v$ must contain only links connecting $v$ to its neighbors (equivalent to edges adjacent to $v$ in $G$) and the parameters $(d_{\min}, d_{\max})$. If the network is anonymous, the process executing the algorithm does not even need to know its identifier. 

The load-balancing part of the algorithm is executed in lines~\ref{line:main-for-start}-\ref{line:main-for-end}. In these lines, we implement a natural counterpart of the classical load balancing process, but adjusted to the presence of faulty nodes. 
This part of the algorithm takes $\tau_1 = 32d^2_{\max} / d^2_{\min} \log{n}$ rounds. This is expected, as it roughly corresponds to the number of iterations a classical error-free random walk would need to converge, in $\| \cdot \|_{1}$ norm, within $\frac{1}{n}$ radius to its limit distribution, if executed on a regularized version of graph $G$.  
The load of a vertex $v$ after round $i \in [\tau_1]$ is stored in variable $x_{i}(v)$; the value $x_{0}(v)$ corresponds to the initial load of the node $v$.

\subsection{Analysis of the main loop of the load balancing algorithm}

We define $x_i$, for $0 \le i \le \tau_1$, as the vector of loads of the nodes after round $i$ (following the above convention, for $i = 0$ the vector $x_{i}$ has entries corresponding to the input values). To analyze the changes in the vectors $\left(x_{i}\right)_{0 \le I \le \tau_1}$ over consecutive rounds, we define three alternatives, which abstract runs of the main loop of the algorithm under different treatment of loads of the faulty nodes. 

\begin{definition}[Ideal load balancing]
The ideal load balancing corresponds to the run of the main loop in which the values of the vector $x_0$ are exchanged along edges of $G^{\cI}$ but assuming that no node is erroneous, where $G^{\cI}$ is the graph obtained from $G$ by adding loops to every vertex such that $G^{\cI}$ is $\gamma d$-regular. Define 
$$x^{\cI}_0 := x_{0}, \text{ and }$$
\[
x_{i}^{\cI}(v) := \left( \sum_{u \in N_{G^{\cI}}(v)} \frac{1}{2\dmax}x_{i - 1}^{\cI}(u) \right) + \frac{2\dmax - |N_{G^{\cI}}(v)|}{2\dmax} x_{i - 1}^{\cI}(v) \quad \forall \quad v \in V, 1 \le i \le \tau_1
\ .
\]
\end{definition}
\begin{definition}[Load balancing skewed toward $0$]
This load balancing corresponds to the run of the main loop of the Algorithm~\ref{alg:load-balancing} in which the loads of faulty nodes are always assumed to be $0$. This skews the convergence towards $0$. Recall that $N_i(v)$ denotes those nodes who send $v$ a load in round $i$ and define 
$$x^{0}_0 := x_{0}, \text{ and }$$
\[
x_{i}^{0}(v) := \left( \sum_{u \in N_{i}(v)} \frac{1}{2\dmax}x_{i - 1}^{0}(u) \right) + \frac{1}{2} x_{i - 1}^{0}(v) \quad \forall \quad v \in V, 1 \le i \le \tau_1
\ .
\] 
\end{definition}
\begin{definition}[Load balancing skewed toward $1$]
This load balancing corresponds to the run of the main loop of the algorithm~\ref{alg:load-balancing} in which the loads of faulty nodes are always assumed to be $1$. This skews the convergence towards $1$. We define 
$$x^{1}_0 := x_{0}, \text{ and }$$
\[
x_{i}^{1}(v) := \left( \sum_{u \in N_{i}(v)} \frac{1}{2\dmax}x_{i - 1}^{1}(u) \right) + \frac{1}{2} x_{i - 1}^{1}(v) + \frac{d_{\min} - |N_{i}(v)|}{2\dmax}  \quad \forall \quad v \in V, 1 \le i \le \tau_1
\ .
\] 
\end{definition}

We start by analyzing the rate of convergence of the ideal load balancing represented by the sequence of vectors $\left( x_{i}^{\cI} \right)_{0\le i \le \tau_1}$. To this end, we define $G^{\cI}$ as the regularized version of the graph $G$ in which every vertex is augmented by such a number of self-loop edges that its degree is $d_{\max}$. Note that adding a self-loop increments the degree of a vertex by $1$.
By the assumption that $G$ is $(d_{\min}, d_{\max})$-\wellconnected, the definition of $G^{\cI}$ is valid. Furthermore, we can utilize the spectral properties of a $(d_{\min}, d_{\max})$-\wellconnected{} graph to infer spectral properties of $G^{\cI}$. 

\begin{lemma}\label{lem:ideal_l2}
The second smallest eigenvalue of the normalized Laplacian of $G^{\cI}$ satisfies: 
$\frac{1}{8} \left( \frac{d_{\min}}{d_{\max}} \right)^{2} \! \le \! \lambda_2\left( G^{\cI} \right)$.
\end{lemma}

By connecting 
results from random walks theory with the previously proven spectral property of $G^{\cI}$ we get: 

\begin{lemma}
\label{lem:classical-random-walk}
Let $\mu = \frac{1}{n}\sum_{v \in G} x_0(v)$. For any $v \in G^{\cI}$, it holds that
$|x^{\cI}_{\tau_1}(v) - \mu| \le \frac{1}{n}$,
where $\tau_1 = 32\frac{d^2_{\max}}{d^2_{\min}} \log{n}$.
\end{lemma}

Next, we show that the entries of the vectors representing the ideal load balancing lie always in-between the entries of the vectors representing load balancing skewed toward $0$ and $1$, respectively.

\begin{lemma}\label{lem:skewed-ideal}
For any $0 \le i \le \tau_1$ and $v \in G_1$, it holds that $x^{0}_i(v) \le x^{\cI}_i(v) \le x^{1}_i(v).$
\end{lemma}

It also holds that the skewed load balancing processes are good approximations of the factual process described by the entries of the vectors $(x_i)_{0\le i \le \tau_1}$.

\begin{lemma}\label{lem:skewed-real}
For any $0 \le i \le \log{n}$ and $v \in G$, it holds that $x^{0}_i(v) \le x_i(v) \le x^{1}_i(v).$
\end{lemma}

Next, we can reason about the entries of the vector $x_{\tau_1}$, which corresponds to the distribution of the loads in the graph vertices after the termination of the main loop of the algorithm.

\begin{lemma}\label{lem:real-main-dist}
For any $\varepsilon \in [\frac{2}{n}, 1]$, there exists a set $C$ of at least $n - \frac{2\tau_1}{\varepsilon} t$ vertices such that for any $v \in C$ it holds that $x_{\tau_1}(v) \in \left[\mu - \varepsilon, \mu + \varepsilon \right].$
\end{lemma}

\subsection{Procedure \textsc{FixOutliers} and final analysis}

By Lemma~\ref{lem:real-main-dist}, there exists a set of vertices $C$ whose loads after the load balancing stage converged to the mean value $\mu$ with the absolute error at most $\varepsilon$. Nevertheless, there is the remaining set of at most $2\tau_1 t / \varepsilon$ vertices which loads can by more than $\varepsilon$ from $\mu$. Next, we analyze the procedure $\textsc{FixOutliers}$, see Algorithm~\ref{alg:fixoutliers}, that resolves this discrepancy by forcing every vertex to have the load within the accepted error or signaling the inability to do so by returning a special flag called $\texttt{type} \in \{active, silent \}$ in the pseudocode. We note that the communication in the procedure $\textsc{FixOutliers}$ is performed on the same $(d_{\min}, d_{\max})$-\wellconnected{} graph $G$ that was used in the load-balancing stage.

We call a vertex \textit{active} if it does not become silent while executing the procedure \textsc{FixOutliers}; otherwise, it is \textit{silent}. 
The set of all active vertices is denoted $A$, while the set of silent ones is denoted~$S$. Using the properties of expanders, we can give a lower bound on the size of the set $A$.

\begin{lemma}\label{lem:active}
If 
$t < \left(\frac{40}{81}\left(\frac{d_{\min}}{d_{\max}}\right)^2 - \frac{2}{9}\frac{d_{\min}}{d_{\max}} \right) n$, then $A$ has size at least $n - \frac{3}{2}t$.
\end{lemma}

We next proceed to establish the number of vertices whose load is within the absolute error of $\varepsilon$ from the mean $\mu$. We will analyze how the number of vertices with this property changes over the rounds of the algorithm~\textsc{FixOutliers}. Among all subsets of $V$, define $C_{0}$ as a subset of maximum size satisfying the following conditions (a) every vertex from $C_0$ has at least \jo{$\frac{1}{2}\left(d_{\max} + 1\right)$} neighbors in $C_{0}$ (b) every vertex in $C_0$ is active, (c) the load of any vertex in $C_0$ differs by at most $\varepsilon$ from $\mu$ before round $1$ starts. We first observe that there exists $C_{0}$ of large size.

\begin{lemma}
\label{lem:active-converged-core}
If $t \le \frac{\varepsilon}{3\tau_1}\left( 1-\frac{d_{\max} + 1}{2d_{\min}}\right)\left(\frac{40}{27}\left(\frac{d_{\min}}{d_{\max}}\right)^2 - \frac{2}{9}\frac{d_{\min}}{d_{\max}} \right)n$, then the size of $C_{0}$ is at least $n - 3\left(\frac{\tau_1}{\varepsilon} + 1\right)t$.
\end{lemma}

We next proceed to analyze how the existence of the set $C_0$ influences the loads of other active vertices. We define $R_{0} = A \setminus C_{0}$. We then define $C_{i}$, for $1 \le i \le \tau_2$, as the set of the vertices that are active and in round $i$ received at least $\frac{1}{2}(d_{\min} + 1)$ messages from vertices in the set $C_{i - 1}$. Let $R_i$ denote the set of active vertices at the end of round $i$ that do not belong to $C_i$, i.e., $R_i = A \setminus C_i$. We show that based on the spectral properties of a \wellconnected\space graph, the size of set $R_i$ diminishes exponentially between consecutive iterations.

\begin{lemma}
\label{lem:rem-shrinks}
If $t < \frac{2d_{\min}}{d_{\max}} \frac{\varepsilon n}{81}$, then for any $1 \le i \le \tau_2$ it holds that  $|R_{i + 1}| < \frac{14}{15}|R_{i}|$.
\end{lemma}

Finally, we show that if a node belongs to $C_{i}$, for some $1 \le i \le \tau_2$, then its load is in $[\mu - \varepsilon, \mu + \varepsilon]$.

\begin{lemma}
\label{lem:final-convinced-converges}
For any node $v \in C_{\tau_2}$ it holds that
$x_{\tau_2}(v) \in [\mu - \varepsilon, \mu + \varepsilon]$.
\end{lemma}

In the next theorem, we summarize the load balancing algorithm. Among properties discussed above, we also highlight that the load of every vertex is always between the smallest and the largest input value of any vertex, as all the arithmetic operations involve either a linear combination of input values with coefficients adding to $1$ or taking a median.   

\begin{theorem}\label{thm:load-balancing}
Let $G$ be a $(d_{\min}, d_{\max})$-\wellconnected\space  graph. 
The algorithm $\textsc{FaultTolerantLLB}$ executed on the graph $G$ under at most $t$ crash or omission failures achieves the following guarantees:\\
\vspace{-3mm}\\
\noindent \textit{$(i)$} it terminates in $O\left(\tau_1 + \tau_2\right)$ rounds using $O\left(d_{\max}\left(\tau_1 + \tau_2
\right)|M|\right)$ communication bits per node, where $|M|$ is the size of machine word, $\tau_1 = 32 \left(\frac{d_{\max}}{d_{\min}} \right)^2 \log{n}$, and $\tau_2 = \log{n} / \log\left(\frac{34}{15} - \frac{4d_{\min}}{3d_{\max}} \right)$;\\ 
\vspace{-3mm}\\ 
\noindent \textit{$(ii)$} the first element of the returned pair is always between the largest and the smallest input value;\\
\vspace{-3mm}\\ 
\noindent \textit{$(iii)$} if $t < \left(\frac{4}{81}\left(\frac{d_{\min}}{d_{\max}} \right)^2 - \frac{2}{9}\frac{d_{\min}}{d_{\max}}\right)n$ then there exists a set of nodes $A$, of size at least $n - \frac{3}{2} t$, such that for every $v \in A$, every returned pair $x(v), \texttt{type}$ satisfies $\texttt{type} = active$;\\
\vspace{-3mm}\\ 
\noindent \textit{$(iv)$} let $\varepsilon \in [0,1]$ be such that $t < \frac{\varepsilon}{3\tau_1}n \cdot f(d_{\min}, d_{\max})$, for some function $f(\cdot, \cdot)$. Then for every node $v \in A$ it holds that 
\vspace*{-0.5ex}
\[
x(v) \in [\mu - \varepsilon, \mu + \varepsilon]
\ ,
\]

\vspace{-2ex}
where $\mu$ denotes the mean of the input values and $f(d_{\min}, d_{\max}) = \left( 1-\frac{d_{\max} + 1}{2d_{\min}}\right)\left(\frac{40}{27}\left(\frac{d_{\min}}{d_{\max}}\right)^2 - \frac{2}{9}\frac{d_{\min}}{d_{\max}} \right)$.
\end{theorem}

\vspace{-2ex}
\begin{proof}
The algorithm's main loop runs for $\tau_1 = 32 \left(\frac{d_{\max}}{d_{\min}} \right)^2 \log{n}$ rounds, with each round corresponding to one communication round. In each of these rounds, every node communicates with at most $d_{\max}$ other processes. Each message sent in this communication consists of at most one machine word. The procedure $\textsc{FixOutliers}$ iterates through another loop for a fixed $O(\tau_2)$ number of times. Each iteration corresponds to one communication round, and again, each node in this round sends a message of at most one machine word to at most $d_{\max}$ other nodes. This proves property $\textit{(i)}$.

To ensure that every value $x(v)$ for any node $v$ remains between the smallest and largest input values, we observe that the values $\left( x_{i}(v) \right)_{0 \leq i \leq \tau_1}$ in line~\ref{line:main-for-end} are updated as a weighted average of a subset of values from the previous round. Since the sum of the weights in this average is $1$, a simple inductive argument shows that for any $i$, the value $x_{i}(v)$ remains within the range of the smallest and largest input values.

In the procedure \textsc{FixOutliers}, nodes use the median of the received values as the new value for the next round. Thus, a similar inductive argument shows that the procedure's output $x(v)$ must also lie between the smallest and largest input values. This completes the proof of property $\textit{(ii)}$.
The existence of the set $A$ follows from Lemma~\ref{lem:active}, which proves property $\textit{(iii)}$.

Finally, property $\textit{(iv)}$ follows from Lemma~\ref{lem:rem-shrinks} and Lemma~\ref{lem:final-convinced-converges}. Lemma~\ref{lem:rem-shrinks} guarantees that after $\tau_{2}$ rounds, the size of $R_{\tau_2}$ is smaller than $1$, implying that $R_{\tau_2} = \emptyset$ and consequently, $C_{\tau_2} = A$. By applying Lemma~\ref{lem:final-convinced-converges}, the property is proven.
\end{proof}


\vspace*{-3ex}
\section{Application of LLB to Counting} 
\label{sec:applications}

\begin{algorithm}[b!]
\LinesNumbered
\SetKwInOut{Input}{input}
\Input{$n$, $flag_{v}$\;}
$N(v) \gets \textsc{SetGraph}(v)$\;
\lIf{$flag_v = true$}{$b_v \gets 1$}
\lElse{$b_v \gets 0$}
$\mu_{v}, \texttt{type}_v \gets \textsc{FaultTolerantLLB}(N(v), (\dmin, \gamma\dmin), v, b_v)$\label{line:load-balancing-counting}\;
\lIf{$\texttt{type}_{v} = active$}{\Return{$\mu_{v}$}}

\caption{\textsc{AlmostEverywhereLLBCounting}\label{alg:counting}}
\end{algorithm}

\begin{algorithm}[ht!]
\LinesNumbered
\SetKwInOut{Input}{input}
\Input{$n$, $v$}
let $p(n)$ be the positive solution to the equation $2x - x^2 = \frac{\jo{C_2}\log{n}\left(\log\log{n}\right)^{2}}{n - 1}$\;
let $N(v)$ be a random subset of $V \setminus \{v\}$ where each element is taken independently with probability~$p(n)$\label{line:sampling-crash}\;
send a dummy message to every $u \in N(v)$\label{line:inquire-set-graph}\;
for every received message from $u$ add $u$ to $N(v)$\label{line:update-set-graph}\;
\Return{$N(v)$}\;
\caption{\textsc{SetGraph}\\
Constant $C_2$ is defined independently in the analysis of each algorithm using \textsc{SetGraph} as subroutine.\label{alg:set-graph}}
\end{algorithm}

\dk{We first show a relatively straightforward}
application of fault-tolerant LLB to the almost-everywhere counting problem. It works for both crash and omission failures and against the adaptive adversary. 
Pseudocode of our solution to the almost-everywhere counting problem is given in Algorithm~\ref{alg:counting}, with instantiation of random links in  Algorithm~\ref{alg:set-graph}. The constant parameters $\dmin$ and $\gamma$ used in the algorithm are defined as follows: $\dmin = \frac{3}{4}\left(2p(n) - p^2(n)\right)$ and $\gamma = \frac{5}{4}$. Note that when executing the algorithm~\textsc{FaultTolerantLLB}, we pass the minimal knowledge needed to evoke this algorithm: $N(v)$ corresponds to the of neighbors of a node~$v$, and $(\dmin, \gamma\dmin)$ describe the parameters of the implicit graph used by the nodes.
Nodes that completed procedure \textsc{FaultTolerantLLB} as active, i.e., the procedure returned $(\mu_v,\texttt{type}_v=active)$, return the value $\mu_v$ as the result of counting.
\dk{The next result immediately follows from Theorem~\ref{thm:load-balancing}:}

\begin{theorem}\label{thm:counting}
For any number of either crash or omission failures $t \le O(\frac{n}{\log{n}})$, the algorithm \textsc{AlmostEverywhereLLBCounting} solves $(n-3t)$-almost-everywhere counting problem, with probability at least $1 - \frac{1}{n^2}$. With the same probability bound, it uses $O(\log{n})$ rounds and $O(n\log^3{n})$ communication~bits.
\end{theorem}

\vspace*{-1ex}
\section{Application of LLB to Consensus with Crash Failures}
\label{sec:app-cons-crashes}

In this section, we use the load-balancing procedure to present a simple Consensus algorithm for the case of at most $n/3$ {\bf\em crash failures}, which recovers the efficiency of other state-of-the-art results. The result is based on the technique proposed in~\cite{Bar-JosephB98}; however, we use the fault-tolerant load-balancing procedure instead of the all-to-all vote-counting scheme of the former.
The algorithm is called $\textsc{LLBConsensus:CrashFailures}$ and is described in Algorithm~\ref{alg:crash-consensus}, with random instantiation of local ports as in Algorithm~\ref{alg:set-graph}.
$Bin(1, 1/2)$ in the code
     denotes Bernoulli random distribution.
\dk{The analysis of the algorithm is deferred to Section~\ref{sec:app-consensus-crash-analysis}.}

Other techniques for replacing the all-to-all vote-counting scheme have been proposed, e.g., the state-of-the-art result of~\cite{hajiaghayi2022improved}; however, fault-tolerant load balancing is much simpler while at the same time outperforming the complexities of the state-of-the-art solution. Specifically, it requires $O(\log{n})$ times fewer rounds and achieves $O(\log^{4}{n}/(\log\log n)^2)$ times better communication complexity, cf. Theorem~2 in~\cite{hajiaghayi2022improved}.

\begin{algorithm}[t!]
\LinesNumbered
\SetKwInOut{Input}{input}
\Input{$n$, $v$, $b_{v}$\;}
$N^{\ast}(v) \leftarrow \textsc{SetGraph}(v)$ \label{alg:set-g-ast-crash} \hfill\tcp{\small $G^{\ast}$ is $((d^{\ast}(1 - 1/\log\log{n}), d^{\ast}(1 + 1/\log\log{n}))$-well-connected}
\For{$i \leftarrow 1 \text{ to } C_1\sqrt{n\log{n}}$} {
    $N_{i}(v) \leftarrow \textsc{SetGraph}(v)$ \hfill\tcp{$G_{i}$ is $(d^{i}_{\min}, d^{i}_{\max})$-well-connected}
    $\mu_{v}, \texttt{lb\_status}_{v} \leftarrow \textsc{FaultTolerantLLB}(N_i(v), d^{i}_{\min}, d^{i}_{\max}, v, b_v)$\vspace{1mm}\label{line:load-balancing}\;
    \For{$i \leftarrow 1 \text{ to } 40\log{n} + 1$ \label{line:dis-start}} {
        send $(\mu_{v}, \texttt{lb\_status}_{v})$ to every vertex in $N^{\ast}(v)$\;
        $M_{v} \leftarrow$ set of messages received by $v$ in the previous round\;
        \lIf{$|M_{v}| < \frac{1}{5}d^{\ast}_{\min}$} {$v$ skips all iterations until line~\ref{line:b-change-0}\label{line:becomes_silent}}
        \lIf{exists $u$ such that $(\mu_{u}, \texttt{lb\_status}_{u} = active) \in M_{v}$} {$\mu_{v}, \texttt{lb\_status}_{v} \leftarrow \mu_{u}, \texttt{lb\_status}_{u}$}\label{line:dis-end}
    }
    \vspace{1mm}
    \lIf{$\mu_{v} < 1/2- \frac{1}{40}\sqrt{\frac{\log{n}}{n}}$ } {$b_{v} \gets 0$ \label{line:b-change-0}}
    \lElseIf{$\mu_{v} > 1/2 + \frac{1}{40}\sqrt{\frac{\log{n}}{n}}$} {$b_{v} \gets 1$\label{line:b-change-1}}
    \lElse{$b_{v} \gets Bin(1, 1/2)$%
    \label{line:random}}
}

\lIf{ever passed line~\ref{line:becomes_silent}}{ask random $10\log{n}$ processes about their variable $b$; upon receiving any response $b_u$ set $b_v \leftarrow b_u$}\label{line:inquiring}
\lElse{response $b_{v}$ to any inquiring process\label{line:responding}}
\Return{$b_{v}$}\;
\caption{\textsc{LLBConsensus:CrashFailures}\label{alg:crash-consensus}\\
Constant $C_1$ and parameters $d^{\ast}, d^{i}_{\min}, d^{i}_{\max}$ are defined in the analysis.
}
\end{algorithm}

\begin{theorem}\label{thm:simple-cons-crashes}
With probability at least $1 - \frac{1}{n}$, algorithm $\textsc{LLBConsensus:CrashFailures}$ solves consensus against $t < n / 3$ crashes in $O(n^{1/2}\log^{3/2}{n})$ rounds using $O(n^{3/2}\log^{5/2}n\cdot (\log\log{n})^2)$ communication~bits.
\end{theorem}

\begin{algorithm}[t!]
\LinesNumbered
\SetKwInOut{Input}{input}
\Input{$n$, $v$, $b_{v}$\;}
$\texttt{type}_{v} \leftarrow active$\label{line:start-active}\;
\For{$i \leftarrow 1 \text{ to } 2C_1 \max\left(\frac{t\log{n}}{\sqrt{n}}, \log{n}\right)$} {
    \lIf{$\texttt{type}_v = suspected$} {skip until line~\ref{line:inq-omissions}\label{line:skip}} 
    $N_{i}(v) \leftarrow \textsc{SetGraph}(n, v)$  \label{line:drawing-omissions}\hfill\tcp{$G_{i}$ is $(d^{i}_{\min}, d^{i}_{\max})$-well-connected}
    $\mu_{v}, \texttt{lb\_status}_v \leftarrow \textsc{FaultTolerantLLB} (N_i(v), d^{i}_{\min}, d^{i}_{\max}, v, b_v)$\label{line:load-balancing-omissions}\;
    $\texttt{omitted}_v \leftarrow $ number of links that at least once failed to deliver a message during the last execution of \textsc{FaultTolerantLLB}\label{line:omitted-def}\;
    \lIf{$\texttt{omitted}_v > 0$} {$\texttt{type}_v \leftarrow suspected$\label{line:omitted-change}}\vspace{2mm}
    \lIf{$\mu_{v} < 1/2- \frac{1}{12}\sqrt{\frac{\log{n}}{n}}$} {$b_{v} \leftarrow 0$\label{line:b-0-omissions}}
    \lElseIf{$\mu_{v} > 1/2 + \frac{1}{12}\sqrt{\frac{\log{n}}{n}}$} {$b_{v} \leftarrow 1$\label{line:b-1-omissions}}
    \lElse{$b_{v} \leftarrow Bin(1, 1/2)$\label{line:random-omissions}}
}
\lIf{$\texttt{type}_{v} = suspected$}{inquire arbitrary $11C_2\log{n}(\log\log{n})^2 t + 1$ processes about the correct value of $b$; set $b_{v}$ to any received response\label{line:inq-omissions}}
\lElse{respond to the requests of the inquiring processes\label{line:res-omissions}}
\Return{$b_{v}$}\;
\caption{\textsc{LLBConsensus:OmissionFailures}\label{alg:omission-consensus}\\
Constant $C_1, C_2$ and parameters $d^{i}_{\min}, d^{i}_{\max}$ are defined in the analysis.
}
\end{algorithm}


\vspace*{-1ex}
\section{Application of LLB to Consensus with Omission Failures}
\label{sec:app-consensus-omissions}

As the next application, we use the~\texttt{FaultTolerantLLB} algorithm to solve consensus in networks with omissions. A pseudocode of the algorithm is given in Algorithm~\ref{alg:omission-consensus} with random instantiation of ports as in Algorithm~\ref{alg:set-graph}. The algorithm tolerates $t < C\frac{n}{\log{n}(\log\log{n})^2}$ omissions failures, for some constant $C < 1$. 

The major difference between this algorithm and its counterpart for crash failures stems from the fact that the pattern of omission failures does not have to be monotonic, i.e., a link can fail to deliver a message in a round to/from a faulty node and reactive later. Thus, we categorize the processes into two sets: active and suspected. At any point of execution, each process is either active or suspected. Initially, all processes start as active, see line~\ref{line:start-active}. A process becomes suspected if in the course of the \texttt{FaultTolerantLLB} procedure there exists at least one neighbor whose message was not received by this process. Adapting the code of the \texttt{FaultTolerantLLB} procedure to change the status of a process from trusted to suspected is straightforward, thus we only mark this new version of algorithm with an \* rather than rewriting the entire pseudocode, see line~\ref{line:load-balancing-omissions}. 
Finally, the introduced labeling of processes is used in the process of drawing the graph, lines~\ref{line:skip}-\ref{line:drawing-omissions}. 

The idea is to exclude suspected processed from the process of deciding consensus value as they have the potential of being failed. \jo{Therefore, once a processes marks itself suspected, it refrains from graph sampling and communication in the iterations of the main loop of the algorithm. This is feasible, because properties of the \textsc{FaultTolerantLLB} guarantee that a substantial fraction of all processes will not turn suspected. The ones that turn suspected, and by this fact are excluded from the knowledge of the decision value, perform an additional round of requesting active processes at the end of the algorithm.}


\begin{theorem}\label{thm:omissions}
\jo{
For some constant $C$, the algorithm $\textsc{LLBConsensus:OmissionsFailures}$ solves Consensus against $t < \frac{n}{C\log{n}\left(\log\log{n}\right)^2}$ omission failures in $O\left(\frac{t\log^2{n}}{\sqrt{n}} + \log^2{n} \right)$ rounds using $O(t^2 \log^2{n}\left(\log\log{n}\right)^4 \allowbreak + \left(t\sqrt{n} + n\right)\log^{3}{n}(\log\log{n})^2)$ communication~bits, with probability at least $1 - \frac{1}{n}$.
}
\end{theorem}

\subsection{Analysis of \textsc{LLBConsensus:OmissionsFailures} and proof of Theorem~\ref{thm:omissions}}

We note that the notation used in this analysis is similar to that introduced in Section~\ref{sec:app-consensus-crash-analysis}, as both algorithms utilize the same core component of the \textsc{FaultTolerantLLB} procedure. We require that $C_1, C_2$  are constants larger than $2^{15}$.\footnote{\jo{We note that while the bound on the constants may seem large, it is chosen mainly to simplify the analysis, in which we aimed at the best asymptotic complexities. By considering communication graphs that are $O(\log{n})$ times denser, these constants can be significantly reduced.}} 
For simplicity, we use $k$ to denote the number of times the main loop of the algorithm is iterated, i.e., $k = 2C_1 \max\left(\frac{t\log{n}}{\sqrt{n}}\right)$. We set the constant $C$ in the formulation of Theorem~\ref{thm:omissions} to $\jo{15 \cdot C_2}$, i.e., \jo{$t < \frac{n}{15C_2\log{n}\left(\log\log{n}\right)^2}$}.

Let $A_{i}$ denote the set of processes that are active at the beginning of iteration $i$ of the main loop algorithm\footnote{Note that, contrary to the definition of $A_i$ in Section~\ref{sec:app-cons-crashes}, an active process can be controlled by the adversary. Also, to simplify the notation, we interpret round $k + 1$ as the end of iteration $i$.}. We define graph $G_{i}$ as the graph formed by the processes in $A_{i}$, taking the union of the sets of edges $\left(N_{i}\right){v \in V}$ at the moment they are drawn in line~\ref{line:sampling-crash} of the procedure \textsc{SetGraph}. For clarity, we introduce the notation $d^{i}_{\min} = \frac{C_2|A_{i}|\log{n}\left(\log\log{n}\right)^2}{n - 1}\left(1 - \frac{1}{20{\log\log{n}}}\right)$ and $d^{i}_{\max} \allowbreak = \frac{C_2|A{i}|\log{n}\left(\log\log{n}\right)^2}{n - 1}\left(1 + \frac{1}{20{\log\log{n}}}\right)$, which denote the expected minimum and maximum degree of $G_i$, as proven in the next Lemma.

We observe that $G_i$ is a theoretical construction, and in practice, the sets $N_i(v)$ passed to the procedure \textsc{FaultTolerantLLB} may be subsets of the corresponding sets in $G_i$. However, this does not affect correctness. Any failure that prevents an edge from $G_i$ from being included in $N_{i}(v)$—see lines~\ref{line:inquire-set-graph}-\ref{line:update-set-graph} in \textsc{SetGraph}—can be assumed to occur in the procedure \textsc{FaultTolerantLLB} before the communication rounds of this procedure begin.
This interpretation does not impact the correctness of \textsc{FaultTolerantLLB}, as in the case of a link failure, correct processes rely solely on the parameters $d^{i}{\min}$, $d^{i}{\max}$, and their own loads to make a transition—see line~\ref{line:main-for-end} in Figure~\ref{alg:load-balancing}. All this information is provided regardless of when a failure happened.
Since we use the same procedure for sampling the graph $G_{i}$ as in the case of crash failures in Section~\ref{sec:app-cons-crashes}, we can state the following result.

\begin{lemma}\label{lem:gi-well-connected-omissions}[Lemma~\ref{lem:gi-well-connected} in Section~\ref{sec:app-consensus-crash-analysis}]
If $|A_{i}| \ge \frac{n}{4}$, for $1 \le i \le k$, then the graph $G_{i}$ is $(d^{i}_{\min}, d^{i}_{\max})$-\wellconnected\space with probability at least $1 - \frac{1}{n^2}$. Furthermore, $\frac{d^{i}_{\max}}{d^{i}_{\min}} \le \frac{11}{10}$.
\end{lemma}

In the following, we assume that for every $1 \leq i \leq k$, the graph $G_i$ satisfies the properties of the above lemma. By the union bound argument, the probability of this event occurring is at least $1 - \frac{1}{\log^{2} n^{3/2}}$. Next, we prove that the set of active processes maintains a sufficiently large size throughout all iterations of the main loop of the algorithm
\begin{lemma}\label{lem:active-large-omissions}
For any $1 \le i \le k + 1$, it holds that $|A_{i}| \ge n - 10C_2\log{n}(\log\log{n})^2 t \ge \frac{n}{4}$. Furthermore, the set $A_{i}$ consists of processes whose links have not failed from the beginning of the execution.
\end{lemma}
\begin{proof}
We first note that the assumption $t < \frac{n}{15C_2\log{n}(\log\log{n})^2}$ implies that $n - 10C_2\log{n}(\log\log{n})^2 t \ge \frac{n}{4}$.
For $1 \leq i \leq k$, let $T_{i}$ denote the set of processes that experienced at least one omission failure during the communication rounds of iteration $i$ of the main algorithm. That is, these are the processes that \textit{failed} to send or receive at least one message during the communication rounds of iteration $i$. We observe that $\sum_{1 \leq i \leq k} |T_{i}| < \frac{n}{C_2 \log{n}(\log\log{n})^2}$.
We will show by induction that $|A_{i}| \geq n - \frac{3}{2} d^{i}_{\max} \sum_{1 \leq j < i} |T_{j}|$. We begin with the base case $i = 1$. Here, the induction hypothesis holds trivially, as no communication occurred before the first iteration. Thus, $n - \frac{3}{2} d^{i}_{\max} |\cup_{1 \leq j < 1} T_{j}| = n$, and obviously, no communication links could have failed.

Consider an iteration $1 < i \leq k + 1$. Any communication rounds that occur in this iteration are invoked by calling either the procedure \textsc{SetGraph} or the procedure \textsc{FaultTolerantLLB}. However, as observed before, without loss of generality, we can treat any failure that occurs in \textsc{SetGraph} as if it happens at the beginning of \textsc{FaultTolerantLLB}. 
Furthermore, following Lemma~\ref{lem:gi-well-connected-omissions}, we condition the analysis on the fact that the graph $G_{i}$ is $(d^{i}_{\min}, d^{i}_{\max})$-\wellconnected, which holds since $|A_{i}| \geq \frac{n}{4}$ by the inductive assumption. 
Therefore, using Theorem~\ref{alg:load-balancing}, point \textit{(ii)}, we conclude that the number of processes that return the variable $\texttt{lb\_status}$ set to true is at least $|A_{i}| - \frac{3}{2}|T_{i}|$. Now, for a process to fail to receive a message in the procedure \textsc{FaultTolerantLLB}, one of the following must be true: either the process itself is faulty, the process on the other side of the link is faulty, or the process on the other side of the link has become inactive during the execution of \textsc{FaultTolerantLLB}; otherwise, the message would be delivered. The number of processes that are faulty or become inactive during \textsc{FaultTolerantLLB} in iteration $i$ is at most 
\[
\frac{3}{2}|T_{i}| + |T_{i}| \leq \frac{5}{2}|T_{i}|,
\]
as previously observed. Since the maximum degree of $G_{i}$ is upper bounded by $d^{i}_{\max}$, at most 
\[
\frac{5}{2}|T_{i}| \cdot d^{i}_{\max} \leq 10C_2\log{n}(\log\log{n})^2 |T_{i}|
\]
processes can fail to receive a message during the communication rounds of iteration $i$. 
Combining this observation with the fact that processes that continuously receive all expected messages remain active (see lines~\ref{line:omitted-def}-\ref{line:omitted-change} of the main algorithm), we obtain
\[
|A_{i + 1}| \geq |A_{i}| - 10C_2\log{n}(\log\log{n})^2 |T_{i}| \geq n - 10C_2\log{n}(\log\log{n})^2 \sum_{1 \leq j \leq i} |T_{j}|,
\]
thus proving the lemma.
\end{proof}

In the next part of the analysis, we proceed with a proof that all active processes store the same value in the variable $b$ upon the termination of the main loop of the algorithm. We note that the techniques used here are the same as those used in the analysis in Section~\ref{sec:app-cons-crashes}. This is inherent due to the exploitation of the \textsc{FaultTolerantLLB} algorithm as the core component of both algorithms. 

We call an iteration of the algorithm's main loop \textit{safe} if at most $\frac{1}{C_1} \frac{\sqrt{n}}{\log{n}}$ processes stop being active in this iteration. Since a process is counted as inactive whenever it fails to receive an incoming message, we conclude that inactive processes encompass the set of erroneous vertices in the procedure \textsc{FaultTolerantLLB}. That is, any process that remains active is indistinguishable from a correct one. Thus, we can state the following.
\begin{lemma}\label{lem:safe-iteration-converge-omissions}
Consider a safe iteration and let $\mu^{\ast}$ denote the average of values $b$ of all process that are active at the beginning of the iteration. If $v$ is active at the end of this iteration, the value of variable $\mu_{v}$ in lines~\ref{line:b-0-omissions}-\ref{line:random-omissions} of the algorithm ~\textsc{LLBConsensus:OmissionFailures} satisfies $\mu_{v} \in \left[\mu^{\ast} - \frac{1}{12}\sqrt{\frac{\log{n}}{n}}, \mu^{\ast} +  \frac{1}{12}\sqrt{\frac{\log{n}}{n}}\right]$.
\end{lemma}
\begin{proof}
Consider a safe iteration $i$. By Lemma~\ref{lem:gi-well-connected-omissions}, the graph $G_i$ is $(d^{i}_{\min}, d^{i}_{\max})$-\wellconnected. Thus, the outcome of the procedure \textsc{FaultTolerantLLB} invoked in line~\ref{line:load-balancing} of the main algorithm is captured by Theorem~\ref{thm:load-balancing}. In particular, simple calculations can verify that the assumptions of points $(iii)$ and $(iv)$ of the theorem are satisfied for the choice $\varepsilon = \frac{1}{12}\sqrt{\frac{1}{n}}$. This is because we are considering a safe iteration meaning that $t < \frac{1}{C_1}\frac{\sqrt{n}}{\log{n}} < \frac{1}{2^{15}}\frac{\sqrt{n}}{\log{n}}$, and also because by Lemma~\ref{lem:gi-well-connected-omissions} we obtain that $\frac{d^{i}_{\max}}{d^{i}_{\min}} \le \frac{11}{10}$. 
In consequence, we obtain the existence of a set $B_i$ of size at least $|A_{i}| - \frac{3}{2C_1}\frac{\sqrt{n}}{\log{n}}$ of the property that every process from this set has the \texttt{lb\_status} variable set to \textit{active}. As only these processes stay active in an iteration of the main loop of the main algorithm, c.f. line~\ref{line:omitted-change}, thus the lemma is proven.
\end{proof}
Establishing that safe iterations output a close enough approximation of the mean of values $b$ at the beginning of the iteration, we can follow the reasoning path as in Section~\ref{sec:app-cons-crashes}.
\begin{lemma}\label{lem:one-all-omissions}
If an iteration $i$ exists such that all active processes store the same value $b$ at the end of the iteration, then all active processes have the same value $b$ at the end of the next iteration.
\end{lemma}
\begin{proof}
Consider the iteration $i+1$. Theorem~\ref{thm:load-balancing}, point \textit{(ii)}, assures that the first value of the output pair of the procedure \textsc{FaultTolerantLLB} is the value $b$ that active processes held at the end of the previous round. That is true for \textit{every} process regardless of its $\textsc{lb\_status}$. This value must be either $0$ or $1$, thus no random choices are made in the final part of the iteration $i+1$ and therefore the lemma follows.
\end{proof}

\begin{lemma}\label{lem:two-iterations-omissions}
Consider two consecutive safe iterations $\cI_{1}, \cI_{2}$. Let $A_{i}$, for $i \in \{1, 2\}$ be the set of active processes at the end of the iteration $\cI_{i}$. With constant probability, all processes belonging to $A_{2}$ store the same value in the variable $b$ at the end of the second iteration.
\end{lemma}
\begin{proof}
We first recall a large deviation inequality.
\begin{lemma}[Lemma $4.3$ in~\cite{Bar-JosephB98}]\label{lem:anti-concetration-omissions}
Assume that $n$ processes independently choose a random bit from uniform distribution. Let $X$ be the random variable denoting the number processes that chose bit $1$. Then for any $t \le \sqrt{n} / 8$
$$\Pr(X - \E(X) \ge t\sqrt{n}) \ge \frac{e^{-4(t+1)^2}}{\sqrt{2\pi}} \ .$$
\end{lemma}
We denote $\mu_{1}, \mu_2$ the average of of values of variables $b$ held by active processes at the beginning of round $\cI_1$ and $\cI_2$, respectively. 
If $\mu_1 > \frac{1}{2} + \frac{1}{6}\sqrt{\frac{\log{n}}{n}}$, then by Lemma~\ref{lem:safe-iteration-converge-omissions} ever active process $v$ has $\mu_{\cI_1}(v) > \frac{1}{2} + \frac{1}{12}\sqrt{\frac{\log{n}}{n}}$ at the end of iteration $\cI_1$. Therefore, it assigns $1$ as the value of the variable $b$. By Lemma~\ref{lem:one-all-omissions}, the variable $b$ is also set to $1$ at the end of the second iteration.
Assume that $\mu_{1} \le \frac{1}{2} + \frac{1}{6}\sqrt{\frac{\log{n}}{n}}$. 
Firstly, we note that in the iteration $\cI_{1}$ no two processes can execute line~\ref{line:b-0-omissions} and line~\ref{line:b-1-omissions} at the same time, since the difference between right-hand-sides of these two inequalities is larger than $\frac{1}{12}\sqrt{\frac{\log{n}}{n}}$ while by Lemma~\ref{lem:safe-iteration-converge-omissions} values $\mu_{\cI_1}(v)$ can differ by at most $\frac{1}{12}\sqrt{\frac{\log{n}}{n}}$. This yields two cases. 

Assume, that no process executes line~\ref{line:b-change-0}. Thus all active processes change the variable $b$ to a uniform random bit or $1$. Since $|A_{1}| \ge \frac{n}{4}$, by Lemma~\ref{lem:active-large-omissions}, thus we can apply Lemma~\ref{lem:anti-concetration-omissions} for $t = \frac{1}{2}$, to conclude that with probability $\frac{1}{10^4}$, more than $\frac{1}{2}|A_{1}| + \frac{1}{4}\sqrt{n}$ processes assign $1$ to their value $b$ at the end of iteration $\cI_1$. This yields
\[
\mu_{2} \ge \left(\frac{1}{2}|A_{1}| + \frac{1}{4}\sqrt{n}\right) / |A_2|  \ge \frac{1}{2} + \frac{1}{6}\sqrt{\frac{1}{n}}, 
\]
where the last inequality follows from the facts that $|A_2| \ge |A_1| - \frac{1}{C_1}\sqrt{n/\log{n}} \ge |A_1| - \frac{1}{C_1}\sqrt{n/\log{n}}$ and  $|A_{1}| \ge n/4$. 
Since iteration $\cI_2$ is also safe, thus by the reasoning analogical to the one presented at the beginning of the lemma, we conclude that all active processes assign $1$ as the value of the variable $b$ at the end of iteration $\cI_2$. 

In the second case, when no process executes line~\ref{line:b-change-0} we reason analogically, but only in this case, we use Lemma~\ref{lem:anti-concetration-omissions} for lower bounding the probability of deviating negatively from the expected value (i.e. we do the estimate for the expected number of $0$'s).
\end{proof}
This lets us finally reason that at the end of the last iteration, all correct processes return the same value with high probability.
\begin{lemma}\label{lem:seq-good-omissions}
With probability $1 - \frac{1}{n}$ at least, all correct processes return the same value $b$ at the end of the algorithm.
\end{lemma}
\begin{proof}
The algorithm iterates the main loop $k = 2 C_1 \max\left(\frac{t\log{n}}{\sqrt{n}}, \log{n}\right)$ many times. Thus, by the pigeonhole principle, there is at least $C_1 \max\left(\frac{t\log{n}}{\sqrt{n}}, \log{n}\right)$ pairs of consecutive safe iterations. For every such pair, with probability at least $\frac{1}{10^4}$, all active processes store the same value in variable $b$ at the end of the second iteration, as per Lemma~\ref{lem:two-iterations-omissions}. Once it happens, Lemma~\ref{lem:one-all-omissions} assures that this value is stored in active processes until the main loop terminates. Now, the probability that every pair of safe iterations fails to unify the variable $b$ across active processes is bounded by 
\[
\left(1 - \frac{1}{10^4}\right)^{C_1 \max\left(\frac{t\log{n}}{\sqrt{n}}, \log{n}\right)} \le \left(\frac{1}{e}\right)^{\frac{C_1}{10^4} \max\left(\frac{t\log{n}}{\sqrt{n}}, \log{n}\right)} \le \frac{1}{n^2},
\]
as $C_1 \ge 2^{15}$.

Let us condition on the event that all active processes have the same value of the variable $b$ upon termination of the main loop of the algorithm. It remains to show that in lines~\ref{line:inq-omissions}-\ref{line:res-omissions} the value stored by the active processes is transmitted to all other processes. As observed in Lemma~\ref{lem:active-large-omissions}, the set of active processes after the main loop end has a size larger or equal to $n - 10C_2\log{n}(\log\log{n})^2 t$. This means that at least $n - 11C_2\log{n}(\log\log{n})^2 t$ is active and non-faulty. 
In line~\ref{line:inq-omissions}, every correct, suspected process sends at least $11C_2\log{n}(\log\log{n})^2 t + 1$ requests, thus by pigeonhole principle, at least one of these requests arrives at a correct, active process. This process must return the appropriate value $b$ to the querying part. On the other hand, all active processes store the same value in variable $b$ as proved at the beginning of this lemma. This guarantees that every correct process eventually returns the same value. 
\end{proof}

\begin{proof}[Proof of Theorem~\ref{thm:omissions}]
We first argue for correctness. The property that all returned values are the same with probability $1 - \frac{1}{n}$ at least follows from Lemma~\ref{lem:seq-good-omissions}. The fact that the value is among the input values from the point $\textit{(ii)}$ of Theorem~\ref{thm:load-balancing}. That is, if all processes receive the same input value, only this value is ever assigned to any variable $b$ in the pseudocode and thus it also must be the decision value.

To derive the round and bit complexity, we observe that the number of iterations of the main loops is fixed to $O\left(\max\left(\frac{t\log{n}}{\sqrt{n}}, \log{n}\right)\right) = O\left( \frac{t\log{n}}{\sqrt{n}} + \log{n}\right)$. By Theorem~\ref{thm:load-balancing}, every single iteration uses $O(\log{n})$ rounds and $O\left(n \log^2(n)(\log\log{n})^2\right)$ communication bits, as the maximum degrees of graphs $G_i$ are universally bounded by Lemma~\ref{lem:gi-well-connected-omissions}. Multiplying the number of iterations by the complexity of a single iteration we get that the round complexity of this part of the algorithm is $O\left( \frac{t\log^2{n}}{\sqrt{n}} + \log^2{n}\right)$, while communication complexity if $O\left( t\sqrt{n}\log^{3}{n}(\log\log{n})^2 + n \log^3(n)(\log\log{n})^2\right)$. 
The inquiring phase takes additional $O(1)$ rounds and $O\left(\log{n}(\log\log{n})^2 t\right)$ communication bits per process. Following Lemma~\ref{lem:active-large-omissions}, there are at most $11C_2\log{n}(\log\log{n})^2 t$ processes that are suspected, thus the total communication complexity of this stage is $O\left(t^2 \log^2{n}(\log\log{n})^4 \right)$. Thus the theorem follows.
\end{proof}

\section{Proofs of technical results from Section~\ref{sec:preliminaries}}
\label{sec:proofs-preliminaries}

\subsection{Proof of Lemma~\ref{cor:edge-density}}

In the proof, we use the following result:

\begin{lemma}[Lemma~1 in~\cite{chung2016generalized}]\label{lem:edge-density-lap}
Denote $\lambda_{2}$ the second smallest eigenvalue of the normalized Laplacian matrix of a graph $G$ and let $S$ be a subset of vertices in $G$. Then
$$
\frac{|\partial(S)|}{\operatorname{vol}(S)} \geq \lambda_2\left(1-\frac{\operatorname{vol}(S)}{\operatorname{vol}(G)}\right) .
$$
\end{lemma}
We now prove Lemma~\ref{cor:edge-density}. By definition, we have the following identity
$$E(S) = \frac{1}{2}\left( vol(S) - |\partial(S)| \right).$$

Plugging the result of Lemma~\ref{lem:edge-density-lap} to the above equality proves the corollary.

\subsection{Proof of Lemma~\ref{lem:high-degree-core}}



Assume to the contrary that the size of the set $W_{k} > \ceil{\frac{3}{2} |F|}$. It follows, that there must be an index $i$, $1 \le i < k$, such that $W_{i} =  \ceil{\frac{3}{2}\beta |F|}$. Observe that when a vertex is added in the construction, it adds at least $(1- \phi)d_{\min}$ edges per step. Since exactly one vertex is added at a step, thus the graph $W_{i}$ has at least 
\begin{align}\label{line:1}
(1-\phi)\left(|W_{i}| - |F|\right) d_{\min} \ge (1-\phi)\left( \ceil{\frac{3}{2} |F|} - |F| \right) d_{\min} \ge \frac{1 - \phi}{2} d_{\min} |F|
\end{align}
internal edges.
On the other hand, using Lemma~\ref{cor:edge-density}, we obtain that the number of internal edges in $W_{i}$ is at most
$$
E(W_i) \le \frac{1}{2}vol(W_{i})\left(1 - \lambda_{2}\left(1 - \frac{vol(W_{i})}{vol(G)} \right) \right) \le \frac{1}{2}\dmax|W_{i}|\left(1 - \lambda_2\left(1 - \frac{\dmax|W_i|}{d_{\min}|V|} \right) \right).
$$
Using the assumption that: (a) $|W_{i}| = \ceil{\frac{3}{2}|F|} \le \alpha\frac{3}{2} |V|$, we can further upper bound the number of internal edges in $W_{i}$ by
\begin{align}\label{line:1.1}
E(W_{i}) \le \frac{3}{4}\dmax |F| \left(1 - \lambda_{2} + \lambda_2\frac{d_{\max}}{d_{\min}}\frac {3}{2} \alpha \right).
\end{align}
Assuming that $n > 4$ and recalling that $G$ is $(\dmin, \dmax)$-\wellconnected, we can use using the following inequalities
\[
\alpha < (1-\phi)\frac{40}{27}\frac{d^2_{\min}}{d^2_{\max}} - \frac{2}{9}\frac{d_{\min}}{d_{\max}}
 \text{ and } \lambda_{2} \ge 1 - \frac{1}{10\log\log{n}} \ge \frac{9}{10}.
\]
to upper bound the the right-hand side inequality in line~\ref{line:1.1}. By simplifying we obtain eventually that
\[
E(W_{i}) < \frac{1-\phi}{2}d_{\min}|F|,
\]
which yields a contradiction with
the inequality in line~\ref{line:1} and proves the lemma.

\subsection{Proof of Lemma~\ref{cor:erdos-renyi-lap}}

We define the expected adjacency matrix of this distribution as $\bar{A} := \E_{G \sim G(n,p)}\left(A(G) \right)$. Similarly, we define $\bar{D} := \E_{G \sim G(n,p)}\left(D(G) \right)$. The definition of an expected normalized Laplacian matrix of $G(n,p)$ extends:
\[
\bar{L} = I - \bar{D}^{-1/2}\bar{A}\bar{D}^{-1/2}
\ .
\]
The following theorem gives probability bounds on the connectivity properties of a random graph in which the edge appears independently with some fixed probability.

\begin{theorem}[Theorem~2 in~\cite{chung2011spectra}]\label{thm:random-graphs-laplacian}
Let $G$ be a random graph where $\operatorname{pr}\left(v_i \sim v_j\right)=p_{i j}$, and each edge is independent of each other edge. Let $A$ be the adjacency matrix of $G$. Let $D$ be the diagonal matrix with $D_{i i}=\operatorname{deg}\left(v_i\right)$, and $\bar{D}=$ $\mathrm{E}(D)$. Let $\delta$ be the minimum expected degree of $G$, and $L=I-D^{-1 / 2} A D^{-1 / 2}$ the (normalized) Laplacian matrix for $G$. Choose $\epsilon>0$. If $\delta>k \ln n$, where $k$ is positive and $k > 3\ln(4n / \varepsilon) / \ln{n}$ then with probability at least $1-\epsilon$, the eigenvalues of $L$ and $\bar{L}$ satisfy
$$
\left|\lambda_j(L)-\lambda_j(\bar{L})\right| \leq 3 \sqrt{\frac{3 \ln (4 n / \epsilon)}{\delta}}
$$
for all $1 \leq j \leq n$, where $\bar{L}=I-\bar{D}^{-1 / 2} \bar{A} \bar{D}^{-1 / 2}$.
\end{theorem}

Using this result we can give a probability bound on the fact that a graph drawn from $G(n,p)$ is 
$(p(n-1)(1 - 1/(20\log\log{n})), p(n-1)(1 + 1/(20\log\log{n})))$-\wellconnected, as stated in Corollary~\ref{cor:erdos-renyi-lap}.

Consider a vertex $v$. The random variable $deg(v)$ follows the Binomial distribution with parameters $(n-1, p)$ and thus the standard Chernoff's bound gives that
\[
\Pr\left[ \left|deg(v) - p(n-1)\right| \ge \frac{p(n - 1)}{20\log\log{n}} \right] \le \exp\big(-\frac{\left(\log\log{n}\right)^{2}}{400 \cdot \left(2 + 1 / \log\log{n}\right)} \cdot p(n-1)\big).
\]
Bounding $1/\log\log{n} < 1$ and using that $p(n - 1) > C\log{n}(\log\log{n})^{2}$, we can set $C = 8 \cdot 400$, and calculate that
\[
\Pr\left[ \left|deg(v) - p(n-1)\right| \ge \frac{p(n - 1)}{20\log\log{n}} \right] \le \exp\big(-4\log{n} \big) = \frac{1}{n^4}.
\]
Taking the union bound over all possible choices of the vertex $v$ and using the complementary rule gives the lower bound $1 - \frac{1}{n^3}$ on the probability of the event $(a)$ being satisfied.

To prove $(b)$, we observe that if $G$ is drawn from $G(n, p)$ then $\bar{D} = p(n-1) \cdot I$ and $\bar{A}$ satisfies $\bar{A}_{i,j} = p$ for every $i,j \in V$, $i \neq j$. This gives that $\bar{L} = I-\bar{D}^{-1 / 2} \bar{A} \bar{D}^{-1 / 2} = \big(1 + \frac{1}{n - 1}\big) I  - \frac{1}{n - 1} \mathbf{1}_{|V|\times|V|}$. The matrix $\big(1 + \frac{1}{n - 1}\big) I$ has $n$ eigenvalues $1 + \frac{1}{n - 1}$ and the matrix $\frac{1}{n - 1} \mathbf{1}_{|V|\times|V|}$ has one eigenvalue $1$ and $n-1$ eigenvalues $0$. This yields that the matrix $\bar{L}$ has one eigenvalue $\frac{1}{n - 1}$ and $n-1$ eigenvalues $1 + \frac{1}{n - 1}$. Therefore, using Theorem~\ref{thm:random-graphs-laplacian} with the parameter $\epsilon$ set to $\frac{1}{2n^2}$, $k = 12$, we get that with probability at least $1 - \frac{1}{2n^2}$ 
\[
\lambda_{2}(G) \ge 1 + \frac{1}{n - 1} - 3\sqrt{\frac{3\ln(4n \cdot 2n^2)}{p(n-1)}} 
 \ge 1 - 1 / (10 \log\log{n}),
 \]
where the last equality holds based on the assumption that $p \ge p(n - 1) > 8\cdot 400\log{n}(\log\log{n})^{2}$. Taking the union bound for the events $(a)$ and $(b)$ completes the proof of Lemma~\ref{cor:erdos-renyi-lap}. 

\section{Proofs of technical results from Section~\ref{sec:load-balancing}}\label{sec:appendix-lb}

\subsection{Proof of Lemma~\ref{lem:ideal_l2}}

We start by providing a definition of edge conductance and restating Cheeger's inequality, following~\cite{notesLap}.
\begin{definition}[Edge Conductance] 
Let $G=(V, E)$ be an undirected graph. The conductance of a subset $S \subseteq V$ and the conductance of the graph $G$ are defined as
$$
\phi(S):=\frac{|\delta(S)|}{\operatorname{vol}(S)} \quad \text { and } \quad \phi(G):=\min _{S: \operatorname{vol}(S) \leq|E|} \phi(S).
$$
\end{definition}

\begin{theorem}[Cheeger's Inequality, Section~3 ~in~\cite{chung1996laplacians}]
Let $G=(V, E)$ be an undirected graph and let $\lambda_2$ be the second smallest eigenvalue of its normalized Laplacian matrix. Then
$$
\frac{1}{2} \lambda_2 \leq \phi(G) \leq \sqrt{2 \lambda_2} .
$$
\end{theorem}
Using Cheeger's inequality for graph $G$, we get that 
\[
\lambda_2(G) \le \phi(G) \le \sqrt{2\lambda_2(G)}.
\]
Observe that by the construction of $G^{\cI}$ the degree of any vertex in $G^{\cI}$ is at most $\frac{d_{\max}}{d_{\min}} \ge 1$ times larger than its degree in $G$. Thus, by the definition of conductance, we obtain that
\[
\phi(G)\frac{d_{\min}}{d_{\max}} \le \phi(G^{\cI}).
\]
On the other hand, applying Cheeger's inequality to $G^{\cI}$ yields
\[
\phi(G^{\cI}) \le \sqrt{2\lambda_2(G^{\cI})}.
\]      
Combining the above with the previous inequalities, we get that 
\[
\lambda_2\left(G^{\cI} \right) \stackrel{}{\ge} \phi\left(G^{\cI}\right)^2 / 2 \stackrel{}{\ge} \left(\phi(G)\frac{d_{\min}}{d_{\max}}\right)^2 / 2 \stackrel{}{\ge} 
\left(\lambda_2(G)\frac{d_{\min}}{d_{\max}}\right)^2 / 2.
 \]
By the fact that $G$ is $(d_{\min}, d_{\max})$-\wellconnected, we get that $\lambda_2(G) = 1 - 1/\left(10\log\log{n}\right)$. Therefore, we can conclude that 
\[
\lambda_2(G^{\cI}) \ge \left(1 - \frac{1}{10\log\log{n}}\right)^2\left(\frac{d_{\min}}{d_{\max}}\right)^2\frac{1}{2} \ge \left(\frac{d_{\min}}{d_{\max}}\right)^2\frac{1}{8},
\]
provided that $n > 4$.

\subsection{Proof of Lemma~\ref{lem:classical-random-walk}}

%
We prove the property by relating the load-balancing process to a random walk on $G^{\cI}$. A random walk on an undirected graph is a random process that chooses a random vertex according to a distribution $p_0$ in and in each consecutive step moves to a neighbor of the current vertex chosen from the uniform distribution of the neighboring vertices (self-loops are counted with multiplicity $1$). A random walk is called \textit{lazy} if the processes choose to sample a neighbor only with probability $\frac{1}{2}$ while with the remaining probability $\frac{1}{2}$ it stays in the current vertex. A random walk on a graph and its normalized adjacency matrix are naturally related as the latter is a compact way of encoding the probabilities of transitions in a step of the random walk. The stationary distribution is the distribution corresponding to the $1$ eigenvector of the normalized adjacency matrix. In the case of a regular graph, the stationary distribution is given by the vector $\pi = \frac{1}{|V|}\mathbf{1}_{|V|}$. Random walks have been studied extensively, in particular, the expected behavior of a random walk and a graph is described by the following theorem.
\begin{theorem}[Theorem 10.4.1, Section~10.5~in~\cite{spielman2025spectral}]\label{thm:markov-conv}
Fix a graph $G$. Let $\lambda_{2}$ denotes the second smallest eigenvalue of its normalized Laplacian matrix. 
Denote $p_t$ the distribution of the position of of a random walk if executed on $G$ for $t$ steps. For all $a, b \in V(G)$ and \( t \), if \( p_0 = \delta_a \), then
\[
\lvert p_t(b) - \pi(b) \rvert \leq \sqrt{\frac{deg_G(b)}{deg_G(a)}} \left(1-\lambda_{2} /2 \right)^t,
\]
where $\delta_a$ denote a distribution assigning mass $1$ to the vertex $a$. 
\end{theorem}
As for the use of the theorem in our case, denote $x = \sum_{v} x^{\cI}_0(v)$ and observe that $\frac{1}{x}x^{\cI}_0$ is a probability distribution. Thus by linearity of matrix multiplication, and by Theorem~\ref{thm:markov-conv}, we obtain that, for any $v \in V$
$|\frac{1}{x}x_{t}^{\cI}(v) - \frac{1}{n}| < \left(1 - \lambda_2/2\right)^t.$
Using the bound from Lemma~\ref{lem:ideal_l2} and setting $t = \tau_1 = 32\left(d_{\max} / d_{\min} \right) ^ 2$ yields
\[
\left|\frac{1}{x}x_{\tau_1}^{\cI}(v) - \frac{1}{n}\right| < \left(1 - \frac{1}{16} \left(\frac{d_{\min}}{d_{\max}} \right)^2\right)^{\tau_1} \le \frac{1}{n^2}
\ ,
\]
where the last inequality follows from the $(1+x)^{r} \le e^{xr}$ inequality.
Multiplying both sides by $x < n$ finally gives
\[
\left|x_{\tau_1}^{\cI}(v) - \frac{x}{n}\right| = \left|x_{t}^{\cI}(v) - \mu\right| \le x \frac{1}{n^2} \le \frac{1}{n}
\ .
\]

\subsection{Proof of Lemma~\ref{lem:skewed-ideal}}

We give a proof by induction. For the base case, we observe that $x^{0}_0 = x^{\cI}_0 = x^{1}_0$ which supports the claim.

Assume that $i \ge 1$ and that the inductive hypothesis holds. The ideal load of a vertex $v \in V$ in step $i$ is defined as
\[
x_{i}^{\cI}(v) = \left( \sum_{u \in N_{G^{\cI}}(v)} \frac{1}{2\dmax}x^{\cI}_{i - 1}(u) \right) + \frac{2\dmax - |N_{G^{\cI}}(v)|}{2\dmax} x^{\cI}_{i - 1}(v)
\ .
\]
Observe that $|N_{G^{\cI}}(v)| = \dmax$ and thus $\frac{2\dmax - |N_{G^{\cI}}(v)|}{2\dmax} = \frac{1}{2}$. On the other hand, it always holds that $N_{i}(v) \subseteq N_{G^{\cI}}(v)$. Using this observation, the inductive hypothesis, and the observation that the loads are always non-negative, we get that
\[
\left( \sum_{u \in N_{G^{\cI}}(v)} \frac{1}{2\dmax}x^{\cI}_{i - 1}(u) \right) + \frac{2\dmax - |N_{G^{\cI}}(v)|}{2\dmax} x^{\cI}_{i - 1}(v) \ge \left( \sum_{u \in N_{i}(v)} \frac{1}{2\dmax} x_{i - 1}^{0}(u) \right) + \frac{1}{2} x^{0}_{i-1}(v) = x^{0}_{i}(v)
\ .
\]
This proves the first inequality of the lemma. 
For the second, we proceed similarly.
We have that
\[
\hspace{-15em} \left( \sum_{u \in N_{G^{\cI}}(v)} \frac{1}{2\dmax}x^{\cI}_{i - 1}(u) \right) + \frac{2\dmax - |N_{G^{\cI}}(v)|}{2\dmax} x^{\cI}_{i - 1}(v)
\]
\[
\hspace{10em} \le \ \left( \sum_{u \in N_{i}(v)} \frac{1}{2\dmax} x_{i - 1}^{1}(u) \right) + \frac{1}{2} x^{1}_{i-1}(v) + \frac{d - |C_{i}(v)|}{2\dmax} \ = \ x^{1}_{i}(v)
\ ,
\]
where the inequality follows from the inductive assumption, the observation that $\frac{2\dmax - |N_{G^{\cI}}(v)|}{2\dmax} = \frac{1}{2}$, and the observation that the loads are always at most $1$ because they are an affine combination of values that are at most $1$. Thus, the lemma is proven.

\subsection{Proof of Lemma~\ref{lem:skewed-real}}

We use induction. For the base case, we observe that $x^{0}_0 = x^{\cI}_0 = x^{1}_0$ which supports the claim.

Assume that $i \ge 1$ and that the inductive hypothesis holds. The real load of a vertex $v \in V$ in step $i$ is defined as
$$x_{i}(v) = \left( \sum_{u \in N_i(v)} \frac{1}{2\dmax}x_{i - 1}(u) \right) + \frac{2\dmax - |N_{i}(v)|}{2\dmax} x_{i - 1}(v).$$

We observe that $N_{i}(v) \subseteq N_{G^{\cI}}(v)$ and thus $2|N_{i}(v)| \le 2\dmax$. Therefore, $\frac{2\dmax - |N_{i}(v)|}{2\dmax} \ge \frac{1}{2}$. Combining this observation with the inductive hypothesis and the fact that loads are always non-negative we get that
\[
\left( \sum_{u \in N_{i}(v)} \frac{1}{2\dmax}x_{i - 1}(u) \right) + \frac{2\dmax - |N_{i}(v)|}{2\dmax} x_{i - 1}(v) \ge \left( \sum_{u \in N_{i}(v)} \frac{1}{2\dmax} x_{i - 1}^{0}(u) \right) + \frac{1}{2} x^{0}_{i-1}(v) = x^{0}_{i}(v).
\]
This proves the left-hand side inequality of the lemma. We use similar reasoning to get the right-hand side inequality. The inductive assumption, the observation that $|N_{i}(v)| \le \dmax$, and the fact that all loads are always at most $1$, as an affine combination of values that are at most $1$, yields that
\[
\hspace{-13em} \left( \sum_{u \in N_{i}(v)} \frac{1}{2\dmax}x_{i - 1}(u) \right) + \frac{2\dmax - |N_{i}(v)|}{2\dmax} x_{i - 1}(v) 
\]
\[
\hspace{10em} \le \left( \sum_{u \in N_{i}(v)} \frac{1}{2\dmax} x_{i - 1}^{1}(u) \right) + \frac{1}{2} x^{1}_{i-1}(v) + \frac{d - |N_{i}(v)|}{2\dmax} = x^{1}_{i}(v).
\]
Thus Lemma~\ref{lem:skewed-real} is proven.

\subsection{Proof of Lemma~\ref{lem:real-main-dist}}

Let $s^1_{i}$ (respectively $s^0_{i}$) denotes the sum of entries of the vector $v^1_i$ (respectively $v^0_i$) for $0 \le i \le \tau_1$. Denote $s^{\cI} = \sum_{v \in G} v^{\cI}_0(v)$. Observe that $s^1_{0} = s^0_{0} = s^{\cI}$. In every round, every edge of the graph $G$ conveys a load of value at most $\frac{1}{2\dmax}$. Since there are at most $t$ erroneous vertices, thus the number of edges that are incident to these vertices is at most $\dmax t$ (we accounted for the loops of every vertex), and in consequence, the absolute difference between $s^1_{i}$ and $s^1_{i - 1}$ can be $t/2$ at most. It follows that $|s^1_{\tau_1} - s^{\cI}| < \tau_1 \cdot t /2$ and, analogously, $|s^0_{\tau_1} - s^{\cI}| < \tau_1 \cdot t / 2$. By Lemma~\ref{lem:skewed-ideal}, $x^1_{\tau_1}$ is larger than $x^{\cI}_{\tau_1}$ on every coordinate and thus we have
\begin{align}\label{line:2}
s^1_{\tau_1}  < s^{\cI} + \tau_1 \cdot t / 2.
\end{align}
By the same lemma, $x^0_{\tau_1}$ is smaller than $x^{\cI}_{\tau_1}$ on every coordinate which yields
$$s^{0}_{\tau_1} > s^{\cI} - \tau_1 \cdot t / 2.$$
Let $C'_{1}$ denote the set of these vertices for which $x^1_{\tau_1}(v) > s^{\cI}_{\tau_1}(n) + \varepsilon/2$. Since for every $v$ it holds that $x^{1}_{\tau_1}(v) \ge x^{\cI}_{\tau_1}(v)$, thus to equality~(\ref{line:2}) to hold it must be that 
$$|C'_{1}| \cdot \varepsilon/2 \le \tau_1 \cdot t/2 \implies |C'_{1}| \le \tau_1 \cdot t/ \varepsilon.$$
Let $C_{1} = V \setminus C'_{1}$. It follows that $C_{1} \ge n - \tau_1 t / \varepsilon$  
$$\forall_{v \in C_{1}} x^1_{\tau_1}(v) \le x^{\cI}_{\tau_1}{n}(v) + \varepsilon / 2 \le \mu + \varepsilon/2 + \frac{1}{n} \ge \mu + \varepsilon,$$
where the last inequality follows from Lemma~\ref{lem:classical-random-walk} and from the assumption that $\varepsilon \ge \frac{1}{n}$.

Applying the same reasoning to the vector $x^0_{\tau_1}$, we obtain a set $C_{0}$ of at least $n - \tau_1 t / \varepsilon$ with the property that
$$\forall_{v \in C_{0}} x^0_{\tau_1}(v) \ge x^{\cI}_{\tau_1}(v) - \varepsilon/2 \ge \mu - \varepsilon.$$
Let $C = C_0 \cap C_1$. We have that $|C| \ge n - 2\tau_1 t / \varepsilon$. Using Lemma~\ref{lem:skewed-real}, combined with the above inequalities, we get that 
$$\forall_{v \in C} \hspace{3mm} \mu - \varepsilon \le x_{\tau_1}(v) \le \mu + \varepsilon,$$
which proves the lemma.

\subsection{Proof of Lemma~\ref{lem:active}}

Let $F$ be the set of the erroneous vertices. By abusing the notation slightly, denote $A$ the remaining graph obtained by applying Lemma~\ref{lem:high-degree-core} on the graph $G$ to the set $F$ for constants $\phi = \frac{2}{3}$ and $\alpha = \left(\frac{40}{81}\left(\frac{d_{\min}}{d_{\max}}\right)^2 - \frac{2}{9}\frac{d_{\min}}{d_{\max}} \right)$ (shortly, we will prove that such defined $A$ is indeed a subset set of active vertices). Since $G$ is $(d_{\min}, d_{\max})$-\wellconnected, and since  $|F| = t  < \left(\frac{40}{81}\left(\frac{d_{\min}}{d_{\max}}\right)^2 - \frac{2}{9}\frac{d_{\min}}{d_{\max}} \right) n$, the assumptions of Lemma~\ref{lem:high-degree-core} are satisfied and we can conclude that size of $A$ is at least $n - \frac{3}{2}|F| \ge n - \frac{3}{2}t$. 
Furthermore, every vertex $v \in A$ has at least $\frac{2}{3}d_{\min}$ neighbors in the subgraph induced by the vertices of $A$. Since $A$ contains only correct vertices, all these $\frac{2}{3}d_{\min}$ neighbors can communicate between each other in all communication rounds of the protocol $\textsc{FixOutliers}$. Therefore, the line~\ref{line:silent} is never executed in any vertex of $A$, which proves that all these vertices must remain $\textit{active}$. 

\subsection{Proof of Lemma~\ref{lem:active-converged-core}}

Let $C$ be the set of these vertices that at the beginning of the algorithm \textsc{FixOutliers} have the load in the interval $[\mu - \varepsilon, \mu + \varepsilon]$. Using Lemma~\ref{lem:real-main-dist}, we get that $|C| \ge n - \frac{2\tau_1}{\varepsilon} t$. 
As the first step of the proof, we define the set $R_{0} = V \setminus (A \cap C)$ as the set of these vertices that at the beginning of the algorithm $\textsc{FixOutliers}$ cannot belong to $C_{0}$. Then, we apply Lemma~\ref{lem:high-degree-core} for $\phi =  \left(\frac{1}{2}d_{\max} + 1\right) / d_{\min}$ to set $R_{0}$ to obtain a core set $C^{\ast}_0$ that also satisfies all the required properties of $C_0$. To apply Lemma~\ref{lem:high-degree-core} we first bound the size of $R_{0}$ 
\[
|R_{0}| \le |V| - |A| + |V| - |C| \le \frac{3}{2}t + \frac{3}{2}\left(\frac{2\tau_1}{\varepsilon} +1\right)t \le \left(\frac{2\tau_1}{\varepsilon} + 2 \right)t
\ ,
\]
where the first inequality follows from the previous lower bound on the size of $C$ and Lemma~\ref{lem:active} that lower bounds the size of $A$. To satisfy the assumptions of Lemma~\ref{lem:high-degree-core} it must hold that
\[
|R_{0}| \le \left(1 - \phi\right) \left(\frac{40}{27}\left(\frac{d_{\min}}{d_{\max}}\right)^2 - \frac{2}{9}\frac{d_{\min}}{d_{\max}} \right) |V|
\ ,
\]
which we verify first using the previously noticed upper bound on the size of $R_0$ by the assumption of this lemma we have that
\[
|R_0| \le  \left(\frac{2\tau_1}{\varepsilon} + 2 \right)t
\ ,
\]
and using the lemma assumption limiting the magnitude of $t$ 
\[
\left(\frac{2\tau_1}{\varepsilon} + 2 \right)t \le \left(\frac{2\tau_1}{\varepsilon} + 2 \right) \cdot \frac{\varepsilon}{3\tau_1}\left( 1-\frac{d_{\max} + 1}{2d_{\min}}\right)\left(\frac{40}{27}\left(\frac{d_{\min}}{d_{\max}}\right)^2 - \frac{2}{9}\frac{d_{\min}}{d_{\max}} \right) 
\]
\[
\le \left(1 - \left(\frac{1}{2}d_{\max} + 1\right) \frac{1}{d_{\min}}\right) \left(\frac{40}{27}\left(\frac{d_{\min}}{d_{\max}}\right)^2 - \frac{2}{9}\frac{d_{\min}}{d_{\max}} \right)
\ .
\]
Therefore, we get the existence of a set $C_0$ of size at least $n - \frac{3}{2}|R_{0}| \ge n - 3\left(\frac{\tau_1}{\varepsilon} + 1 \right)t$ with the following properties (a) every vertex in $C_0$ has degree at least $\frac{1}{2}d_{\max} + 1$ in $C_0$; (b) every vertex in $C_0$ is active since $V - R_{0} \subseteq A$; (c) every vertex in $C_0$ has the load in the interval $[\mu - \varepsilon, \mu + \varepsilon]$ as it holds that $V - R_{0} \subseteq C$.

\subsection{Proof of Lemma~\ref{lem:rem-shrinks}}

We first prove that $R_{i + 1} \subseteq R_{i}$ for any $0 \le i \le \log_{14/15}{n} + 1$. To this end, we observe that $C_{i} \subseteq C_{i + 1}$. Indeed, by definition, every vertex belonging to $C_{i}$, for any $0 \le i \le \log_{14/15}{n} + 1$, must be active. Therefore, if process $v$ receives at least $\frac{1}{2}\left(d_{\max} + 1\right)$ messages from $C_{i-1}$ in a round $i$, it all receives messages from these set of neighbors in any round $j > i$ and thus it belongs to $C_{j}$ for any $j \ge i$. Since $R_{i} = A \setminus C_{i}$, thus the inclusion $R_{i + 1} \subseteq R_{i}$ follows.

By Lemma~\ref{cor:edge-density} the number of internal edges of $R_i$ is at most
\[
E(R_i) \le \frac{1}{2}vol(R_{i})\left(1 - \lambda_{2}\left(1 - \frac{vol(R_i)}{vol(V)} \right) \right).
\]
Using the the fact that $G$ is $(d_{\min}, d_{\max})$-\wellconnected\space, we can upper bound $vol(R_i)$ and lower bound $\lambda_2$ with $\lambda_2 \ge \frac{9}{10}$ yielding
\[
\frac{1}{2}vol(R_{i})\left(1 - \lambda_{2}\left(1 - \frac{vol(R_i)}{vol(V)} \right) \right)
\le \frac{1}{2}d_{\max} |R_{i}|\left(1 - \frac{9}{10}\left(1 - \frac{d_{\max} |R_{i}|}{d_{\min} n} \right) \right).
\]
Using the observation that $R_{i + 1} \subseteq R_{i}$, and Lemma~\ref{lem:active-converged-core} we derive that $|R_{i}| \le |R_{0}| = |V| - |C_0| \le 3\left(\frac{\tau_1}{\varepsilon} + 1 \right)t$ and thus it follows
\[
\frac{1}{2}d_{\max} |R_{i}|\left(1 - \frac{9}{10}\left(1 - \frac{d_{\max} |R_{i}|}{d_{\min} n} \right) \right) 
\le 
\frac{1}{2}d_{\max} |R_{i}|\left(\frac{1}{10} + \frac{9}{10}\frac{d_{\max}}{d_{\min}} 3\left(3\frac{\tau_1}{\varepsilon} + 1\right)\frac{t}{n} \right) \]
\begin{align}\label{line:9}
= \frac{1}{20}d_{\max}|R_{i}| + \frac{27}{20} \frac{d^2_{\max}}{d_{\min}}\left(\frac{\tau_1}{\varepsilon} + 1 \right)\frac{t}{n} |R_i|.
\end{align}
Observe that the assumption that $t < \frac{2d_{\min}}{d_{\max}} \cdot \frac{\varepsilon n}{81}$ implies that $27 \frac{d_{\max}}{d_{\min}}\left(\frac{\tau_1}{\varepsilon} + 1 \right)t < \frac{n}{3}$. By plugging this inequality to the left hand side of~\ref{line:9}, we get finally, that
\[
E(R_i) \le \frac{1}{20}d_{\max}|R_{i}| + \frac{27}{20} \frac{d^2_{\max}}{d_{\min}}\left(\frac{\tau_1}{\varepsilon} + 1 \right)\frac{t}{n} |R_i| \le \frac{1}{15}d_{\max}|R_i|.
\]
On the other hand, we argue that the number of edges by which messages are sent to vertices from $R_{i}$ in round $i+1$ is at least $\frac{2}{3}d_{\min}$. That follows from the inspection of line~\ref{line:if-silent} of the procedure~\textsc{FixOutliers}. Since $R_{i} \subseteq A$, thus in round $i + 1$ every vertex maintains its active status and thus receives at least $\frac{2}{3}d_{\min}$ messages. 
Thus, the number of edges between $C_{i}$ and $R_{i}$ that transmit a message in round $i + 1$ is at least
\[
\left(\frac{2}{3}d_{\min} - \frac{2}{15}d_{\max}\right)|R_{i}.
\]
Using the averaging argument, we conclude that there must be at least.
\[
2\left(\frac{2d_{\min}}{3d_{\max}} - \frac{2}{15} - \frac{1}{2} \right) |R_{i}| = \left(\frac{4d_{\min}}{3 d_{\max}} - \frac{19}{15} \right) |R_{i}|.
\]
vertices of $R_{i}$ who receive $\frac{1}{2}d_{\max} + 1$ or more messages from $C_{i}$ in round $i + 1$ via edges that transmit messages in round $i+1$. By definition of the set $C_{i+1}$, it follows that all these vertices belong to $C_{i+1}$ and thus we have that
\[|R_{i + 1}| \le \left(1 - \frac{4d_{\min}}{3 d_{\max}} + \frac{19}{15} \right)|R_{i}| = \left(\frac{34}{15} - \frac{4d_{\min}}{3 d_{\max}} \right)|R_i|,
\]
which proves the lemma.

\subsection{Proof of Lemma~\ref{lem:final-convinced-converges}}

We reason by induction. Specifically, we prove that for any  $i$ in the range $0 \le i \le \tau_{2}$ and any vertex $v \in C_i$ it holds that $x_{i}(v) \in [\mu - \varepsilon, m + \varepsilon]$. The base case $i = 0$ follows immediately from Lemma~\ref{lem:active-converged-core}.

For $i \ge 1$, observe that any vertex $v \in C_{i}$ receives at least $\frac{1}{2}d_{\max} + 1$ messages from vertices in $C_{i}$, by definition. Then the load of $v$ in round $i$ is set to the median of all received values in these messages. First, the maximum degree of any vertex in $C_{i}$ is at most $d_{\max}$. Second, at least $\frac{1}{2}d_{\max} + 1$ received values is in the interval $[\mu - \varepsilon, m + \varepsilon]$ by the inductive assumption on $C_{i}$. Henceforth, the median of the received values must be also in $[\mu - \varepsilon, m + \varepsilon]$ and thus the lemma follows.

\section{Analysis of Algorithm \textsc{LLBConsensus:CrashFailures}
and Proof of Theorem~\ref{thm:simple-cons-crashes}}
\label{sec:app-consensus-crash-analysis}

We first explain the notation used in the algorithm. The constants $C_{1}, C_{2}$ must be greater than $2^{15}$\footnote{Similarly as in Section~\ref{sec:app-consensus-omissions}, we made no effort to optimize this constant. In particular, increasing the degree of the random graph by $\log{n}$ decreases the requirement on the magnitude of $C_1, C_2$ by an order of magnitude ($\sim2^8$), albeit at the cost of higher asymptotic complexities.} We assume that $n$ is the only parameter given to the algorithm, and all other expressions are either absolute constants or derived based on $n$. We say that a process remains \textit{active} in execution of the algorithm if it is non-faulty and never passes the conditional statement in line~\ref{line:becomes_silent}. The remaining processes are called \textit{silent}. Analogously, we say that a process remains active until iteration $i$ of the main loop of the algorithm if it neither crashes nor passes the conditional statement in rounds $1, \ldots, i$. 

We begin the analysis by noting the properties of the graph formed by the union of edges belonging to the sets $\{ N^{\ast} \}_{v \in V}$.
\begin{lemma}\label{lem:g-ast-well-connected}
Let $G^{\ast} = \cup_{v \in V} N^{\ast}(v)$. Then with probability at least $1 - \frac{1}{n^2}$ graph $G^{\ast}$ is $(d^{\ast}(1-(1/20\log\log{n})), \allowbreak d^{\ast}(1+1/(20\log\log{n})))$ \wellconnected, where $d^{\ast} = C_2\log{n}\left(\log\log{n}\right)^{2}$. Also the ratio of the maximum vertex degree to the minimum vertex degree in $G^{\ast}$ is at most $\frac{11}{10}$.
\end{lemma}
\begin{proof}
Let us note that graph $G^{\ast}$ is drawn from a distribution $G(n, 2p(n) - p(n)^2) \allowbreak = \allowbreak G\left(n, \frac{C_2\log{n}\left(\log\log{n}\right)^2}{n - 1}\right)$ as per the definition of $p(n)$ in line~\ref{line:sampling-crash} of Algorithm~\ref{alg:set-graph}. Thus, by Lemma~\ref{cor:erdos-renyi-lap} we get that $G^{\ast}$ is $(d^{\ast}(1-(1/20\log\log{n})), d^{\ast}(1+1/(20\log\log{n})))$-\wellconnected\space with probability at least $1 - \frac{1}{n^2}$ as we defined $d^{\ast} = C^{\ast}\log{n}\left(\log\log{n}\right)^2$. This give the first part of the lemma. For proving the second part, the upper bound, we note
\[
\frac{d^{\ast}(1+1/(20\log\log{n}))} {d^{\ast}(1-1/(20\log\log{n}))} \le 1 + 
\frac{2 / (20\log\log{n})}{1 - 1/(20\log\log{n})}
\]
\[
\le 1 + \frac{1}{10\log\log{n}} \le \frac{11}{10},
\]
where for the last inequality we used the fact that $n \ge 3$.
\end{proof}

We define $d^{\ast}_{\min} = d^{\ast}(1-1/(20\log\log{n}))$ and $d^{\ast}_{\max} = d^{\ast}(1+1/(20\log\log{n}))$.
Also, in the remainder, we condition on the fact that $G^{\ast}$ is as described in the above lemma. Next, we observe that at least $\frac{1}{4}n$ processes is active.
\begin{lemma}\label{lem:active-are-large}
If $t < \frac{n}{3}$, then at in any execution at least $\frac{n}{4}$ processes remains active, whp.
\end{lemma}
\begin{proof}
By Lemma~\ref{lem:g-ast-well-connected}, graph $G^{\ast}$ must be $(d^{\ast}_{\min}, d^{\ast}_{\max})$~\wellconnected\space and $\frac{d^{\ast}_{\max}}{d^{\ast}_{\min}} \le \frac{11}{10}$.

Let $F$ denote the set of processes that crash in the execution of the algorithm. 
Denote $\alpha = \frac{1}{3}$ and $\phi = \frac{1}{5}$. Following the assumptions of Lemma~\ref{lem:high-degree-core}, we verify that
\[
(1- \phi)\frac{40}{27}\left(\frac{d^{\ast}_{\min}}{d^{\ast}_{\max}} \right)^2 - \frac{2}{9}\frac{d^{\ast}_{\min}}{d^{\ast}_{\max}} \ge \frac{4}{5}\left(\frac{10}{11}\right)^2 - \frac{2}{9} \ge \alpha = \frac{1}{3}.
\]
Thus by applying Lemma~\ref{lem:high-degree-core} to set $F$ on graph $G^{\ast}$, with the choice of parameters $\alpha = \frac{1}{3}$, $\phi = \frac{1}{5}$, we obtain the existence of a set of processes $C$ of size $n - \frac{3}{2}|F| \ge \frac{1}{4}n$, disjoint with $F$ with the property, that every vertex in $C$ has at least $\frac{1}{5}d^{\ast}_{\min}$ neighbors in $C$.
As only correct processes constitute to the set $C$, we observe that none of these processes can pass the conditional statement of line~\ref{line:becomes_silent} in the main algorithm thus proving that all these processes are active.
\end{proof}

In the next lemma, we proof that the set of active processes can exchange any information by taking part in the communication rounds of the inner loop of an iteration of the main loop of the algorithm, e.g. lines~\ref{line:dis-start}-\ref{line:dis-end}. We say that an active process $p$ \textit{reaches} an active process $q$ if there exists a path of reliable links that could transmit a message from $p$ to $q$ during communication in the mentioned lines during execution of these lines.
\begin{lemma}\label{lem:g-ast-broadcast}
If processes $p$ and $q$ remain active in an iteration of the main loop of the algorithm, then $p$ reaches $q$.
\end{lemma}
\begin{proof}
Consider the process $q$. For $1 \le  40\log{n}$, we denote $R^{i}(q)$ the set of active processes who can reach $q$ in $i$ communication rounds of the inner loop in an iteration of the main loop of the algorithm. By induction, we will show that $|R^{i}(q)| \min\left((1 + 1/20)^i, n/20\right)$. The base case follows from the fact that $q$ never executes line~\ref{line:becomes_silent}, and thus in any round of the inner loop receives $\frac{1}{5}d^{\ast}_{\min} > 1 + 1/20$ messages at least.

We now proceed to the inductive step. Consider the set $R^{i}(q)$. If $|R^{i}(q)| \ge n/20$, then the case is proven. Assume then that $|R^{i}(q)| < n/20$. Since all processes in $R^{i}(q)$ are active, thus they must receive at least $\frac{1}{5}d^{\ast}_{\min}$ messages from other active processes in every iteration of the inner loop. The number of links delivering this messages is at least $\frac{1}{5}d^{\ast}_{\min}|R^{i}(q)|$. On the other hand, by Lemma~\ref{cor:edge-density} and the fact that $G^{\ast}$ is $(d^{\ast}_{\min}, d^{\ast}_{\max})$-\wellconnected, we observe that the number of edges internal to $R^{i}(v)$ can be bounded as follows.
\[
E(R^{i}(v)) \le \frac{1}{2}\cdot vol(W)\left(1 - \lambda_{2}\cdot\left(1 - \frac{vol(W)}{vol(G)} \right)\right) \le \frac{1}{2}d^{\ast}_{\min}|R^{i}(v)|\left(1 - \frac{9}{10}\left(1 - \frac{10}{11}\frac{1}{20} \right) \right) \le \frac{1}{14}d^{\ast}_{\min}|R^{i}(v)|.
\]
Therefore, at least $\frac{1}{5}d^{\ast}_{\min}|R^{i}(q)| - \frac{2}{14}d^{\ast}_{\min}|R^{i}(v)| \ge \frac{2}{35}d^{\ast}_{\min}|R^{i}(q)|$ of the aforementioned links connect $R^{i}(v)$ with $G^{\ast} \setminus R^{i}(v)$. Given that the maximum degree of $G^{\ast}$ is $d^{\ast}_{\max} \le \frac{11}{10}d^{\ast}_{\min}$, we conclude that these links are incident to at least $\frac{10}{11}\cdot\frac{2}{35}|R^{i}(v)|$ processes that do not belong to $R^{i}(v)$. It yields that $R^{i + 1}(v) \ge (1 + \frac{1}{20})R^(i)(v)$ and thus the inductive hypothesis is proven. Since the inner loop iterates fro $40\log{n} + 1$ steps, thus we get that $|R^{40\log{n}}(v)| \ge \frac{n}{20}$. 

Observe, that the same reasoning applies to $p$ yielding another set of at least $\frac{n}{20}$ active processes that can reach $p$. Now, using the Expander Mixing Lemma (c.f. Lemma 2.5 in~\cite{hoory2006expander}), we can deduce that there is at least one edge between processes that can reach $p$ and those that can reach $q$, since $G^{\ast}$ satisfies that $\lambda_2 \ge 1-1/(10\log\log{n})$. Since both these sets are formed by active vertices, thus this edge serves as a reliable link between the corresponding sets in all iterations of the inner loop, ultimately proving that $p$ reaches $q$.
\end{proof}

Establishing that a constant fraction of processes remain active and information can be transmitted quickly between active processes via lines~\ref{line:dis-start}-\ref{line:dis-end}, we proceed to analyzing how the variables $b$ change in the active processes during an execution. 
Let $A_{i}$ be the set of these processes that are active at the beginning of iteration $i$ of the main loop algorithm. We define graph $G_{i}$ as the graph formed on the processes in $A_{i}$ by taking the union of the sets of edges $\left(N_{i}\right)_{v \in V}$ at the moment they are drawn in line~\ref{line:sampling-crash} of the procedure \textsc{SetGraph}. For clarity, we introduce the notation $d^{i}_{\min} = \frac{C^{\ast}|A_{i}|\log{n}\left(\log\log{n}\right)^2}{n - 1}\left(1 - \frac{1}{20{\log\log{n}}}\right)$ and $d^{i}_{\max} \allowbreak = \frac{C^{\ast}|A_{i}|\log{n}\left(\log\log{n}\right)^2}{n - 1}\left(1 + \frac{1}{20{\log\log{n}}}\right)$ that denotes the expected minimum and maximum degree of $G_i$ as proven below.
\begin{lemma}\label{lem:gi-well-connected}
For any $1 \le i \le C_1 \sqrt{n\log{n}}$, the graph $G_{i}$ is $(d^{i}_{\min}, d^{i}_{\max})$-\wellconnected with probability $1 - \frac{1}{n^2}$ at least. Furthermore, $\frac{d^{i}_{\max}}{d^{i}_{\min}} \le \frac{11}{10}$.
\end{lemma}
\begin{proof}
By Lemma~\ref{lem:active-are-large}, $|A_{i}| \ge \frac{n}{4}$. Observe that $G_i$ is drawn from the distribution $G\left(|A_i|, 2p(n) - p(n)^2 \right) = G\left(|A_i|, \frac{C_2\log{n}\left(\log\log{n}\right)^2}{n - 1} \right)$. Since $C_2$ is more than four times large than the equivalent constant in Lemma~\ref{cor:erdos-renyi-lap}, and $4|A_i| \ge n$, thus we can apply Lemma~\ref{cor:erdos-renyi-lap} and conclude that $G_{i}$ is $(d^{i}_{\min}, d^{i}_{\max})$-\wellconnected\space with probability at least $1 - \frac{1}{n^2}$. The derivation of the upper bound on the ratio of the maximum to the minimum degree of $G_i$ follows exactly the same argument as provided in the proof of Lemma~\ref{lem:g-ast-well-connected}.
\end{proof}
In the remainder, we assume all graphs $G_{i}$ drawn in the algorithm are $(d^{i}_{\min}, d^{i}_{\max})$-\wellconnected. By Lemma~\ref{lem:gi-well-connected} and by the union bound argument, probability of this event is at least $1 - \frac{\log^{2}(n)}{n^{3/2}}$.


We call an iteration of the algorithm's main loop \textit{safe} if at most $\frac{1}{C_1}\sqrt{n / \log{n}}$ processes stop being active in this iteration. Observe that a crashed processor cannot be active, thus this bound implies that at most $\frac{1}{C_1}\sqrt{n / \log{n}}$ crashes happens in this iteration.
In our proof, we follow the argument that in a sequence of safe iterations, the load-balancing process invoked in line~\ref{line:load-balancing} returns an accurate enough proportion of $1$'s to the total number of vertices that suffices for all processes to assign the same value to the variable $b$ with a constant probability.
\begin{lemma}\label{lem:safe-iteration-converge}
Consider a safe iteration and let $\mu^{\ast}$ denote the average of values $b$ of all process that are active at the beginning of the iteration. If $v$ is active at the end of this iteration, the value of variable $\mu_{v}$ in lines~\ref{line:b-change-0}-\ref{line:random} of the algorithm ~\textsc{LLBConsensus:CrashFailures} satisfies $\mu_{v} \in [\mu^{\ast} - \frac{1}{40}\sqrt{\frac{\log{n}}{n}}, \mu^{\ast} +  \frac{1}{40}\sqrt{\frac{\log{n}}{n}}]$.
\end{lemma}
\begin{proof}
Consider a safe iteration $i$. By Lemma~\ref{lem:gi-well-connected}, the graph $G_i$ is $(d^{i}_{\min}, d^{i}_{\max})$-\wellconnected. Thus, the outcome of the procedure \textsc{FaultTolerantLLB} invoked in line~\ref{line:load-balancing} of the main algorithm is captured by Theorem~\ref{thm:load-balancing}. In particular, simple calculations can verify that the assumptions of points $(iii)$ and $(iv)$ of the theorem are satisfied for the choice $\varepsilon = \frac{1}{40}\sqrt{\frac{\log{n}}{n}}$. This is because we are considering a safe iteration meaning that $t < \frac{1}{C_1}\sqrt{n/\log{n}} < \frac{1}{2^{15}}\sqrt{n/\log{n}}$, and also because by Lemma~\ref{lem:gi-well-connected} we obtain that $\frac{d^{i}_{\max}}{d^{i}_{\min}} \le \frac{11}{10}$. 
In consequence, we obtain the existence of a set $B_i$ of size at least $|A_{i}| - \frac{3}{2C_1}\sqrt{n/\log{n}}$ of the property that every processes from this set has the \texttt{lb\_status} variable set to \textit{active}. 
Reiterating the assumption that at most $\frac{1}{C_1}\sqrt{n/\log{n}}$ processes turn silent in this iteration, we conclude that at least $|A_{i}| - \frac{5}{2C_1}\sqrt{n / \log{n}}$ processes from the set $B_i$ remains active until at the end of iteration $i$, denoted $D_i$. As $|A_{i}| \ge \frac{n}{4}$, by Lemma~\ref{lem:active-are-large}, thus $|A_{i}| - \frac{5}{2C_1}\sqrt{n / \log{n}} > 1$ and therefore $D_i$ is not empty. Applying Lemma~\ref{lem:g-ast-broadcast} to the set $D_i$ and the remaining set of active processes, we conclude that any active process receives the value $\mu_{v}$ of an active process that terminated the \textsc{FaultTolerantLLB} algorithm with the \texttt{lb\_status} variable set to \textit{active}. By Theorem~\ref{thm:load-balancing} it is irrelevant which processes value is acquired as all of them lie in the interval $\left[\mu^{\ast} - \frac{1}{40}\sqrt{\frac{\log{n}}{n}}, \mu^{\ast} +  \frac{1}{40}\sqrt{\frac{\log{n}}{n}}\right]$ and the lemma is proven.
\end{proof}

The remaining part of the proof shows that if the accuracy in the process of counting the fraction of values $b$ equal $1$ to all active processes is precise enough, then there exists a lower bound on the  probability that all active processes converge to the same value $b$ due to the randomness invoked in the line~\ref{line:random}. We remark that this technique has been well-established, c.f.~\cite{Bar-JosephB98, hajiaghayi2024nearly} in the context of synchronous consensus algorithm and we do not claim this part as the main contribution of our paper. We include the proofs for the sake of self-completeness of the paper. 

We start by noting a simple fact that once a system converged to the same value of variable $b$ stored by all active processes, this state is maintained until the end of the main loop of the algorithm.

\begin{lemma}\label{lem:one-all}
If an iteration $i$ exists such that all active processes store the same value $b$ at the end of the iteration, then all active processes have the same value $b$ at the end of the next iteration.
\end{lemma}
\begin{proof}
Consider the iteration $i+1$. Theorem~\ref{thm:load-balancing}, point \textit{(ii)}, assures that the first value of the output pair of the procedure \textsc{FaultTolerantLLB} is the value $b$ that active processes held at the end of previous round. That is true for \textit{every} process regardless of its $\textsc{lb\_status}$. Thus in lines~\ref{line:dis-start}-\ref{line:dis-end} only this value circulates among active processes. This value must be either $0$ or $1$, thus no random choices are made in the final part of the iteration $i+1$ and thus the lemma follows.
\end{proof}
 
\begin{lemma}\label{lem:two-iterations}
Consider two consecutive safe iterations $\cI_{1}, \cI_{2}$. Let $A_{i}$, for $i \in \{1, 2\}$ be the set of active processes at the end of the iteration $\cI_{i}$. With probability $\sqrt{\frac{\log{n}}{4n}}$ all processes belonging to $A_{2}$ store the same value in the variable $b$ at the end of second iteration.
\end{lemma}
\begin{proof}
We first recall a large deviation inequality.
\begin{lemma}[Lemma $4.3$ in~\cite{Bar-JosephB98}]\label{lem:anti-concetration}
Assume that $n$ processes independently choose a random bit from uniform distribution. Let $X$ be the random variable denoting the number processes that chose bit $1$. Then for any $t \le \sqrt{n} / 8$
$$\Pr(X - \E(X) \ge t\sqrt{n}) \ge \frac{e^{-4(t+1)^2}}{\sqrt{2\pi}} \ .$$
\end{lemma}
We denote $\mu_{1}, \mu_2$ the average of of values of variables $b$ held by active processes at the beginning of round $\cI_1$ and $\cI_2$, respectively. 
If $\mu_1 > \frac{1}{2} + \frac{1}{20}\sqrt{\frac{\log{n}}{n}}$, then by Lemma~\ref{lem:safe-iteration-converge} ever active process $v$ has $\mu_{\cI_1}(v) > \frac{1}{2} + \frac{1}{40}\sqrt{\frac{\log{n}}{n}}$ at the end of iteration $\cI_1$. Therefore, it assigns $1$ as the value of the variable $b$. By Lemma~\ref{lem:one-all}, the variable $b$ is also set to $1$ at the end of the second iteration.
Assume that $\mu_{1} \le \frac{1}{2} + \frac{1}{20}\sqrt{\frac{\log{n}}{n}}$. 
Firstly, we note that in the iteration $\cI_{1}$ no two processes can execute line~\ref{line:b-change-0} and line~\ref{line:b-change-1} at the same time, since the difference between right-hand-sides of these two inequalities is larger than $\frac{1}{20}\sqrt{\frac{\log{n}}{n}}$ while by Lemma~\ref{lem:safe-iteration-converge} values $\mu_{\cI_1}(v)$ can differ by at most $\frac{1}{20}\sqrt{\frac{\log{n}}{n}}$. This yields two cases. 

Assume, that no process executes line~\ref{line:b-change-0}. Thus all active processes either change the variable $b$ to an uniform random bit or $1$. Since $|A_{1}| \ge \frac{n}{4}$, by Lemma~\ref{lem:active-are-large}, thus we can apply Lemma~\ref{lem:anti-concetration} for $t = \frac{1}{8}\log{n}$, to conclude that with probability $\frac{1}{16}\sqrt{\frac{\log{n}}{n}}$, more than $\frac{1}{2}|A_{1}| + \frac{1}{16}\sqrt{n\log{n}}$ processes assign $1$ to their value $b$ at the end of iteration $\cI_1$. This yields
\[
\mu_{2} \ge \left(\frac{1}{2}|A_{1}| + \frac{1}{16}\sqrt{n\log{n}}\right) / |A_2|  \ge \frac{1}{2} + \frac{1}{20}\sqrt{\frac{\log{n}}{n}}, 
\]
where the last inequality follows from the facts that $|A_2| \ge |A_1| - \frac{1}{C_1}\sqrt{n/\log{n}} \ge |A_1| - \frac{1}{C_1}\sqrt{n/\log{n}}$ and  $|A_{1}| \ge n/4$. 
Since iteration $\cI_2$ is also safe, thus by the reasoning analogical to the one presented in the beginning of the lemma, we conclude that all active processes assign $1$ as the value of the variable $b$ at the end of iteration $\cI_2$. 

In the second case, when no process executes line~\ref{line:b-change-0} we reason analogically, but only in this case we use Lemma~\ref{lem:anti-concetration} for lower bounding the probability of deviating negatively from the expected value (i.e. we do the estimate for the expected number of $0$'s).
\end{proof}
Finally, we can prove the main theorem justifying correctness of the algorithm. In short, we use the pigeonhole principle two guarantee a large number of pairs of consecutive safe iterations. Since each pair has lower bounded probability
\begin{lemma}\label{lem:seq-good}
With probability at least $1 - \frac{1}{n}$, all correct processes return the same value $b$ at the end of the algorithm.
\end{lemma}
\begin{proof}
Recall, that the adversary can crash at most $n/3$ vertices. The algorithm iterates the main loop $C_1\sqrt{n\log{n}}$ many times. Thus, the pigeonhole principle, there is at least $\frac{C_1}{3}\sqrt{n\log{n}}$ pairs of consecutive safe iterations. For every such pair, with probability at least $\sqrt{\frac{\log{n}}{4n}}$, all active processes store the same value in variable $b$ at the end of the second iteration, as per Lemma~\ref{lem:two-iterations}. Once it happens, Lemma~\ref{lem:one-all} assures that this value is stored in active processes until the main loop terminates. Now, the probability that every pair fails to unify the variable $b$ across active processes is bounded by 
\[
\left(1 - \sqrt{\frac{\log{n}}{4n}}\right)^{C_1\sqrt{n\log{n}}} \le \left(\frac{1}{e}\right)^{C_1\sqrt{n\log{n}}\cdot \sqrt{\frac{\log{n}}{4n}}} = \left(\frac{1}{e}\right)^{C_1 / 16 \log{n}} \le \frac{1}{2n^2}
\ ,
\]
as $C_1 \ge 2^{15}$.

Let us condition on the event that all active processes have the same value of the variable $b$ upon termination of the main loop of the algorithm. It remains to show that in lines~\ref{line:inquiring}-\ref{line:responding} the value stored by active process is transmitted to all other processes. Consider a correct correct process $v$. It submits $10\log{n}$ random requests. Since the number of active processes is at least $\frac{n}{4}$, thus Chernoff's bound implies that with probability at least $1-\frac{1}{2n^2}$ at least one of its requests hits an active process. As active processes are also non-faulty, it is guaranteed that the inquiring process receives a response with the variable $b$ of the active process. Finally, the union bound argument provides that with probability at least $1 - \frac{1}{n}$ all non-faulty processes get a response from an active processes, thus the lemma is proven. 
\end{proof}

\begin{proof}[Proof of Theorem~\ref{thm:simple-cons-crashes}]
We first argue for correctness. The property that all returned values are the same follows from Lemma~\ref{lem:seq-good}. The fact that the value is among the input values from the point $\textit{(ii)}$ of Theorem~\ref{thm:load-balancing}. That is, if all processes receive the same input value, only this value is ever assigned to any variable $b$ in the pseudocode and thus it also must be the decision value.

To derive the round and bit complexity, we observe that the number of iterations of the main loops is fixed to $O(\sqrt{n\log{n}})$ and, by Theorem~\ref{thm:load-balancing}, every single iteration uses $O(\log{n})$ rounds and $O(n \log^2{n})$ communication bits. Finally, the inquiring phase takes additional $O(1)$ rounds and $O(n\log{n})$ communication bits.
\end{proof}

\section{Conclusions and Open Problems}

We developed an efficient implementation of deterministic LLB against adaptive crashes and omissions, by using a specific LLB formula together with fixing outliers. We demonstrated substantial improvements when applying our algorithm, initiated by random links, to solve selected problems of counting and consensus (the resulting algorithms are randomized, due to the random links' initialization of the LLB sub-routine). 

The most promising open direction include further applications of LLB to other distributed computing problems, as well as an attempt to extend this technique to other types of failures, in particular, different types of Byzantine faults. Further shrinking the (poly)logarithmic gaps in formulas for the number of tolerated failures, and for the time and communication complexities, is a natural challenge.

\bibliographystyle{plain}
\bibliography{bibliography}

\end{document}